\documentclass[letterpaper, 10 pt, conference]{ieeeconf}  
\IEEEoverridecommandlockouts                              
\usepackage{graphics} 
\usepackage{epsfig} 
\usepackage{amsmath} 
\usepackage{amsfonts}
\usepackage{amssymb}  
\usepackage{comment}
\usepackage{multicol}
\usepackage[thmmarks]{ntheorem}
\usepackage{cite}
\def\coloneqq{\mathrel{\mathop:}=}

\title{\LARGE \bf
Distributed Output Regulation of Heterogeneous Uncertain Linear Agents
}

\author{Satoshi Kawamura$^{1}$, Kai Cai$^{1}$, and Masako Kishida$^{2}$
	\thanks{$^{1}$Satoshi Kawamura and Kai Cai are with Department of Electrical and Information Engineering, Osaka City University, Japan.
		{\tt\small Emails: kawamura@c.info.eng.osaka-cu.ac.jp, kai.cai@eng.osaka-cu.ac.jp}
	}
	\thanks{$^{2}$Masako Kishida is with Principles of Informatics Research Division, National Institute of Informatics, Japan.
		{\tt\small Email: kishida@nii.ac.jp}
	}
	\thanks{This work was supported by the open collaborative research program at National Institute of Informatics (NII) Japan (FY2018)
		and the Research and Development of Innovative Network Technologies to Create the Future of National Institute of Information and Communications Technology (NICT) of Japan.}
}

\theoremheaderfont{\normalfont\bfseries}
\theorembodyfont{\normalfont}
\newtheorem{definition}{Definition}
\newtheorem{theorem}{Theorem}
\newtheorem{lemma}{Lemma}
\newtheorem{problem}{Problem}
\newtheorem{assumption}{Assumption}
\newtheorem{remark}{Remark}

\begin{document}
	\maketitle
	\thispagestyle{empty}
	\pagestyle{empty}

	\begin{abstract}
	We study a multi-agent output regulation problem, where not all agents have access to the exosystem's dynamics.
	We propose a fully distributed controller that solves the problem for linear, heterogeneous, and uncertain agent dynamics as well as time-varying directed networks.
	The distributed controller consists of two parts: (1) an {\it exosystem generator} that locally estimates the exosystem dynamics,
	and (2) a {\it dynamic compensator} that, by locally approaching the internal model of the exosystem, achieves perfect output regulation.
	Moreover, we extend this distributed controller to solve an output synchronization problem where not all agents initially have the same internal model dynamics.
	Our approach leverages methods from internal model based controller synthesis and multi-agent consensus over time-varying directed networks; the derived result is a generalization of the (centralized) internal model principle to the distributed, networked setting.
\end{abstract}

	\section{INTRODUCTION}
\label{sec:intro}
	Over the past decade, many distributed control problems of networked multi-agent systems have been extensively studied; these include e.g. consensus, averaging,
	synchronization, coverage, and formation (e.g. \cite{RenBeard:08, BulCorMar:09, LinWanHanFu:14, FB-LNS, mesbahi2010graph}). Progressing beyond first/second-order
	and homogeneous agent dynamics, the {\it distributed output regulation problem} with general linear (time-invariant, finite-dimensional) and heterogeneous agent
	dynamics has received much recent attention (e.g. \cite{WanHonHuaJia:10,SuHua:TAC12,SuHonHua:13,LiuHuang:15,cai2015cooperative,YanDiscrete17}). In this problem, a network of agents each tries to match its
	output with a reference signal, under the constraint that only a few agents can directly measure the reference. The reference signal itself is typically generated by
	an external dynamic system, called ``exosystem''.  The distributed output regulation problem not only subsumes some earlier problems such as (leader-following) consensus and synchronization, but also addresses issues of disturbance rejection and robustness to parameter uncertainty.
	Also see e.g. \cite{SuHua:AUT13,8264161} for further extensions of this problem that deal with nonlinear agent dynamics.
	
	Output regulation has a well-studied centralized version: A single plant tries to match its output with a reference signal (while maintaining the plant's internal stability)
	\cite{Dav:76,FraWon:76,Fra:77}. In the absence of system parameter uncertainty, the solution of the ``regulator equations'', embedding a copy of the exosystem dynamics,
	provides a solution to output regulation \cite{Fra:77}.  When system parameters are subject to uncertainty, however, a dynamic compensator/controller must be used embedding $q$-copy of
	the exosystem, where $q$ is the number of (independent) output variables to be regulated. The latter is well-known as the {\it internal model principle} \cite{FraWon:76}.
	These methods for solving the centralized output regulation problem, however, cannot be applied directly to the distributed version, inasmuch as not all agents have direct access
	to the reference signal or the exosystem dynamics.

	The distributed output regulation of networks of heterogeneous linear agents is studied in \cite{SuHua:TAC12}. The proposed distributed controller consists of two parts:
	an exosystem generator and a controller based on regulator equation solutions. Specifically, the exosystem generator of each agent aims to (asymptotically) synchronize with
	the exosystem using consensus protocols, thereby creating a local estimate of the exosystem. Meanwhile each agent independently tracks the signal of its local generator, by applying
	standard centralized methods (in \cite{SuHua:TAC12} regulator equation solutions are applied). This approach effectively separates the controller synthesis into two parts -- distributed exosystem
	generators by network consensus and local output regulation by regulator equation solution.  

	One important limitation, however, of the above approach is: in both the exosystem generator design and the regulator equation solution, it is assumed that each agent uses
	exactly the same dynamic model as that of the exosystem. This assumption may be unreasonable in the distributed network setting, because those agents that cannot directly measure the reference
	signal are unlikely to know the {\it precise} dynamic model of the exosystem. To deal with this challenge, \cite{cai2015cooperative} proposes (in the case of static networks)
	an ``adpative'' exosystem generator and an adaptive solution to the regulator equations. In essence, each agent runs an additional consensus algorithm to update their
	``local estimates'' of the exosystem dynamics. 

	All the regulator-equation based solutions above fall short in addressing the issue of system parameter uncertainty. In practice one may not have precise knowledge of some
	entries of the system matrices, or the values of some parameters may drift over time. The distributed output regulation problem considering parameter uncertainty is studied in 
	\cite{WanHonHuaJia:10,SuHonHua:13}. The proposed controller is based on the internal model principle, but does not employ the two-part structure mentioned above. It appears to
	be for this reason that restrictive conditions (acyclic graph or homogeneous nominal agent dynamics) have to be imposed in order to ensure solving output regulation.
	Moreover, it is also assumed in \cite{WanHonHuaJia:10,SuHonHua:13} that each agent knows the exact model of the exosystem dynamics.  
	 
	In this paper, we provide a new solution to the distributed output regulation problem of heterogeneous linear agents, where the agents do not have an accurate dynamic model of the exosystem and the agent dynamics are subject to parameter uncertainty.
	In this setting, to our best knowledge, no solution exists in the literature.
	In particular, we propose to use the two-part structure of the distributed controller in the following manner: The first part is an  exosystem generator 
	that works over time-varying networks (\cite{LiuHuang:16,Cai:arxiv16})\footnote{This was apparently developed in \cite{LiuHuang:16} and in \cite{Cai:arxiv16} independently. The first versions of \cite{LiuHuang:16} and \cite{Cai:arxiv16} appeared on arXiv.org, with the former three months earlier than the latter. We thank Dr. Liu and Dr. Huang for in a correspondance bringing our attention to their work.}, and the second part is a dynamic compensator embedding an internal model of the exosystem that addresses parameter uncertainty.
	The challenge here is, in the design of the dynamic compensator, those agents that cannot directly measure the exosystem have {\it no} knowledge of the internal model of the exosystem;
	on the other hand, we know from \cite{FraWon:76} that a precise internal model is necessary to achieve perfect regulation with uncertain parameters.
	To deal with this challenge, we propose a novel consensus-based local internal model for each agent to estimate the internal model of the exosystem. For this time-varying local internal model, we moreover design novel strategies for its eigenvalues to avoid certain transmission zeros of the agents' dynamics in order to guarantee the existence of a dynamic compensator for all time.
	In addition, we extend our new solution to solve a related problem of output synchronization \cite{scardovi2009synchronization, wieland2009internal, wieland2011internal, lunze2012synchronization, grip2015synchronization, seyboth2015robust, Zuo2017}. In this problem there is no exosystem; yet the outputs of all agents are required to converge to the same (dynamic) values.

	The contributions of this paper are threefold. First, the proposed internal-model based distributed controller is the first solution to the multi-agent output regulation problem where the agents do not know {\it a priori} the internal model ($q$-copy) of the exosystem and the agent dynamics are uncertain. Concretely, the proposed distributed controller provably solves the multi-agent output regulation problem in which the following constraints/conditions simultaneously hold: (a) Unknown dynamic model of the exosystem. This is not considered in \cite{WanHonHuaJia:10,SuHonHua:13}. (b) Parameter uncertainty of agent dynamics. This is not addressed in \cite{SuHua:TAC12,cai2015cooperative}. (c) Non-minimum phase agent dynamics. This is not dealt with in \cite{LiuHuang:15}. (d) Time-varying directed networks. This is not addressed in \cite{cai2015cooperative,YanDiscrete17,SuHua:AUT13,8264161}. (e) Heterogeneous agent dynamics. This is not dealt with in \cite{YanDiscrete17}. Second, our solution to the output synchronization problem improves the literature \cite{scardovi2009synchronization, wieland2009internal, wieland2011internal, lunze2012synchronization, grip2015synchronization, seyboth2015robust, Zuo2017} by providing capability of dealing with uncertain agent dynamics, and not requiring all agents initially to have the same internal model dynamics. These improvements allow easier implementation of the proposed controller in a distributed setting.
	As a third contribution, the core of our solution is the time-varying local internal model ($q$-copy), updated in the network setting, which is in itself new in the literature of the internal model principle (cf. \cite{FraWon:76,Fra:77,knobloch2012topics,huang2004nonlinear}) and generalizes the (static, centralized) internal model to the dynamic, distributed one.
	
	In addition we note that \cite{LiuHuang:16} proposes a distributed controller to solve the consensus problem whose design idea is similar to ours.  We point out, however, a few important differences. First, the consensus problem is different from the output regulation problem (the former is usually viewed as a special case of the latter with full-state observation). Second, while  \cite{LiuHuang:16} deals with a class of nonlinear systems, the eigenvalues of the exosystem are required to be distinct. We do not make such an assumption; thus (i) the set of signals that can be generated by the exosystem is a strict superset of that in \cite{LiuHuang:16}, and (ii) the minimal polynomial of the exosystem is generally different from the characteristic polynomial. Third, our designed distributed controller is based on the internal model principle, which is different from the controller designed in \cite{LiuHuang:16}. Finally, while the parameter uncertainty considered in \cite{LiuHuang:16} is represented by a vector, the uncertainty  in this paper is represented by matrices.
	
	The rest of the paper is organized as follows.
	Section~\ref{sec:prelim} introduces the concept of communication graphs and formulates the robust output regulation problem.
	Section~\ref{sec:structuredist} presents the solution distributed controller, which consists of two parts -- a distributed exosystem generator and a distributed dynamic compensator.
	Section~\ref{sec:main} states our main result and provides its proof.
	In Section~\ref{sec:moregeneral} we design the more general distributed controller which addresses non-minimum phase agent dynamics with purely imaginary transmission zeros.
	Section~\ref{sec:synchro} extends our proposed controller to solve an output synchronization problem.
	Section~\ref{sec:example} illustrates our result by simulation examples. Finally, Section~\ref{sec:conclusion} states our conclusions.
	\footnote{The conference version of this paper has been submitted to ACC'19. This paper improves the conference version in the following aspects. (i) A new problem of output synchronization is studied and solved by extending our controller design (Section~\ref{sec:main}). (ii) Elaborated simulations are provided including output regulation over a large-scale network and an output synchronization example. (iii) Detailed analyses and proofs are provided.}

	\section{PRELIMINARIES}
\label{sec:prelim}
	In this paper, we will use the following notation.
	Let ${\bf 1}_n\coloneqq[1\cdots1]^\top\in\mathbb{R}^n$,
	and $I_n$ be the $n \times n$ identity matrix.
	For a complex number $c\in\mathbb{C}$, denote its complex conjugate by $c^*$.
	Write $\mathbb{C}_+$ for the closed right half (complex) plane;
	$\sigma(A)$ for the set of all eigenvalues of $A$.
	We say that a (square) matrix is stable if the real parts of all its eigenvalues are negative.	

\subsection{Agents and Exosystem}
	We consider a network of $N$ agents that are linear, time-invariant, and finite-dimensional. The dynamics of each agent $i(=1,\dots,N)$ is given by
	\begin{align}
		&\dot{x}_i=A_ix_i+B_iu_i+P_iw_0 \label{eq:statetransision}\\
		&z_i=C_ix_i+D_iu_i+Q_iw_0\label{eq:outputtransision}
	\end{align}
	where $x_i\in\mathbb{R}^{n_i}$ is the state vector, $u_i\in\mathbb{R}^{m_i}$ the control input, $z_i\in\mathbb{R}^{q_i}$ the output to be regulated, and $w_0\in\mathbb{R}^r$
	the {\it exogeneous signal} generated by the {\it exosystem}
	\begin{align}
		\dot{w}_0=S_0w_0.
		\label{eq:exosystem}
	\end{align}
	Here $A_i, B_i, C_i, D_i, P_i, Q_i$ and $S_0$ are real matrices of appropriate sizes.
	The signal $w_0$ represents reference to be tracked and/or disturbance to be rejected: $P_iw_0$ in (\ref{eq:statetransision})
	represents disturbance acting on the agent $i$'s dynamics and $Q_iw_0$ in (\ref{eq:outputtransision}) represents reference signals to be tracked by agent $i$.
	
	\begin{assumption}
		\label{ass:apriori}
		The exosystem's $w_0$ and $S_0$ are not (initially) known by the $N$ agents.
	\end{assumption}
	
	Note that the agents are generally heterogeneous: Each of the matrices $A_i, B_i, C_i, D_i, P_i$ and $Q_i$ may have different dimensions and entries.
	Furthermore, we consider that the matricies may have uncertainty;
	namely
	\begin{align}
		&A_i = A_{i0}+\Delta A_i,\ B_i = B_{i0}+\Delta B_i,\ C_i = C_{i0}+\Delta C_i,\nonumber\\
		&D_i = D_{i0}+\Delta D_i,\ P_i = P_{i0}+\Delta P_i,\ Q_i = Q_{i0}+\Delta Q_i\label{eq:perturbation}
	\end{align}	
	where $A_{i0},B_{i0},C_{i0},D_{i0},P_{i0},Q_{i0}$ are the nominal parts of agent $i$ and $\Delta A_i,\Delta B_i,\Delta C_i,\Delta D_i,\Delta P_i,\Delta Q_i$ are the uncertain parts.
	These uncertainty parts may represent measurement errors in the actual determination of the physical parameters, or the reality that these parameters may change with time due to wear and aging \cite{knobloch2012topics}.
	
\subsection{Communication Digraphs}
	Given a multi-agent system with $N(\geq1)$ agents and an exosystem, we represent the time-varying interconnection among the agents and the exosystem by a digraph
	$\hat{\mathcal{G}}(t)=(\hat{\mathcal{V}},\hat{\mathcal{E}}(t))$, where $\hat{\mathcal{V}}=\mathcal{V}\cup\{0\}$, $\mathcal{V}=\{1,\dots,N\}$, is the node set, and $\hat{\mathcal{E}}(t)\subseteq\hat{\mathcal{V}}\times\hat{\mathcal{V}}$ is the edge set. 
	The node $i$, $i=1,\dots,N$, represents the $i$th agent, and the node 0 the exosystem.
	Moreover, $\hat{\mathcal{V}}$ is the node set including the exosystem and $\mathcal{V}$ is the node set except for the exosystem.
	The $i$th node receives information from the $j$th node at time $t$ if and only if $(j,i)\in\hat{\mathcal{E}}(t)$.
	We consider the digraph $\hat{\mathcal{G}}(t)$ does not contain selfloop edges, i.e. $(i,i) \notin \hat{\mathcal{E}}(t)$ for all $i \in \hat{\mathcal{V}}$.
	Only those nodes $i\in\mathcal{V}$ such that $(0,i)\in\hat{\mathcal{E}}(t)$ can receive information from the exosystem 0 (i.e. $w_0$, $S_0$) at time $t$.
	The {\it union digraph} for a time interval $[t_1,t_2]$ is defined as $\hat{\mathcal{G}}([t_1,t_2])\coloneqq(\hat{\mathcal{V}}, \cup_{t\in[t_1,t_2]}\hat{\mathcal{E}}(t))$.
	\begin{definition}
		The digraph $\hat{\mathcal{G}}(t)$ {\it uniformly contains a spanning tree} if there is a finite $T>0$ such that for every $t\geq0$ the union digraph $\hat{\mathcal{G}}([t,t+T])$
		contains a spanning tree.
	\end{definition}
	
	Further, we need the following notions.
	Consider a union digraph $\mathcal{G}([t_1, t_2])=(\mathcal{V}, \cup_{t\in[t_1,t_2]}\mathcal{E}(t))$ (excluding the exosystem). Let $\mathcal{V}_r\subseteq\mathcal{V}$ be a nonempty subset of $\mathcal{V}$. Then the digraph $\mathcal{G}_r([t_1,t_2])=(\mathcal{V}_r,\mathcal{E}_r([t_1,t_2]))$, where $\mathcal{E}_r([t_1,t_2])\coloneqq(\mathcal{V}_r\times \mathcal{V}_r)\cap(\cup_{t\in[t_1,t_2]}\mathcal{E}(t))$, is said to be the {\it induced subdigraph} of $\mathcal{G}([t_1, t_2])$ by $\mathcal{V}_r$.
	
	\begin{definition}
		A {\it strongly connected component} $\mathcal{G}_r([t_1,t_2]) = (\mathcal{V}_r, \mathcal{E}_r([t_1,t_2]))$ of a union digraph $\mathcal{G}([t_1,t_2]) = (\mathcal{V}, \cup_{t \in [t_1,t_2]} \mathcal{E}(t))$ is a maximal induced subdigraph of $\mathcal{G}([t_1,t_2])$ by $\mathcal{V}_r$ which is strongly connected. Moreover, $\mathcal{G}_r([t_1,t_2])$ is a {\it closed} strongly connected component if for every $i \in \mathcal{V}_r$ and every $j \in \mathcal{V} \setminus \mathcal{V}_r$, $(j,i) \notin \mathcal{E}_r([t_1,t_2])$.
	\end{definition}	
	\begin{definition}		
		Consider the time-varying digraph $\mathcal{G}(t)=(\mathcal{V},\mathcal{E}(t))$ and let $\mathcal{V}_r \subseteq \mathcal{V}$ be nonempty. We say that $\mathcal{G}(t)$ {\it uniformly contains a spanning tree with respect to $\mathcal{V}_r$} if there is a finite $T>0$ such that for every $t\geq0$ the union digraph $\mathcal{G}([t,t+T])$ contains a spanning tree and there is a unique closed strongly connected component (subdigraph) induced by $\mathcal{V}_r$.
	\end{definition}
	
	We define the communication weight $a_{ij}(t)$ by
	$a_{ij}(t)\geq\epsilon$ (where $\epsilon$ is a positive constant) if $(j,i)\in\hat{\mathcal{E}}(t)$, and $a_{ij}(t)=0$ if $(j,i)\notin\hat{\mathcal{E}}(t)$.	We assume that $a_{ij}(t)$ is piecewise continuous and bounded for all $t\geq0$ (a technical assumption to be used in Lemma~\ref{lem:converge} below).
	Note that the exosystem does not receive information from any agents, and thus $a_{0j}(t)=0$ for all $j\in\mathcal{V}$, $t\geq0$.
	
	For time $t\geq0$ and digraph $\hat{\mathcal{G}}(t)$, the graph Laplacian $L(t)=[l_{ij}(t)]\in\mathbb{R}^{(N+1)\times(N+1)}$ is defined as
	\begin{align*}
		l_{ij}(t)\coloneqq\left\{
			\begin{array}{ll}
				\sum_{j=0}^N a_{ij}(t),&i=j\\
				-a_{ij}(t),&i\neq j
			\end{array}\right.			
	\end{align*}
	where $i,j\in\{0,\dots,N\}$.	
	
\subsection{Problem Statement}
	We represent by $\hat{\mathcal{G}}(t)=(\hat{\mathcal{V}},\hat{\mathcal{E}}(t))$ the time-varying interconnection among the $N$ agents and the exosystem as in the preceding subsection. In particular, at any time only a subset of agents (possibly different across time) can receive information from the exosystem. This differs the current problem from the traditional,
	centralized output regulation problem \cite{FraWon:76,Fra:77,knobloch2012topics,huang2004nonlinear}. Even if an agent receives information from the exosystem at some time, the agent does not know whether the information is from the exosystem or another agent. Namely we consider that the agents do not have the numbering information including the exosystem (numbered 0).
	\begin{problem}[Distributed Output Regulation Problem]
	\label{prob:mainproblem}
		Given a network of agents (\ref{eq:statetransision}), (\ref{eq:outputtransision}), (\ref{eq:perturbation}) and an exosystem (\ref{eq:exosystem}) with interconnection represented by $\hat{\mathcal{G}}(t)$ and with Assumption~\ref{ass:apriori},
		design for each agent $i\in\mathcal{V}$ a distributed controller such that
		\begin{align*}
			\lim_{t\to\infty} z_i(t)=0
		\end{align*}
		for all $x_i(0), w_0(0)$.
	\end{problem}
	
In the next section we solve Problem~1 by designing an internal-model based distributed controller.	

\subsection{Motivating Example}
	\label{sec:CompEx}
	Reference \cite{cai2015cooperative} considers Problem~\ref{prob:mainproblem} but without the uncertainty part in (\ref{eq:perturbation}), and proposes an effective solution based on regulator equations (for time-invariant digraphs). However, this solution cannot deal with uncertain agent dynamics, as we shall illustrate by an example.
	
	Consider the time-invariant network as displayed in Fig.~\ref{fig:ex_graph_2015}. The exosystem (node~0) is
	\begin{align*}
		\dot{w}_0(t)=S_0w_0,\quad S_0=\left[\begin{array}{cc}0 & 2\\ -2 & 0\end{array}\right].	
	\end{align*}
	The agents $i(=1,2,3,4$) are
	\begin{align*}
		&\dot{x}_i=A_ix_i+B_iu_i+P_iw_0\\
		&z_i=C_ix_i+D_iu_i+Q_iw_0
	\end{align*}
	where
	\begin{align*}
		&A_i = A_{i0}+\Delta A_i,\ B_i = B_{i0}+\Delta B_i,\\
		&A_{i0}=\left[\begin{array}{ccc}0 & 1 & 0\\ 0 & 0 & 1 \\ i & 0 & 0\end{array}\right],\
		B_{i0}=\left[\begin{array}{c}1\\ 0\\ (0.5+0.1i)^2\end{array}\right],\\
		&\Delta A_i = \left[\begin{array}{ccc}0 & 0.5 & 0\\ 0 & 0 & 0 \\ -0.5 & 0 & 0\end{array}\right],\
		\Delta B_i = \left[\begin{array}{c}0.5\\ 0\\ 0\end{array}\right],\\
		&C_i=\left[1\ 0\ 0\right],\
		D_i=0,\
		P_i=\left[\begin{array}{cc}0 & 1\\ 0 & 0\\ 0 & 1\end{array}\right],\
		Q_i=\left[-1\ 0\right].
	\end{align*}
	The initial states $w_0(0),\ x_i(0)$ are selected uniformly at random from the interval $[-1,1]$.
	\begin{figure}[t]
		\begin{minipage}{0.27\hsize}
			\centering
				\includegraphics[scale=1.0]{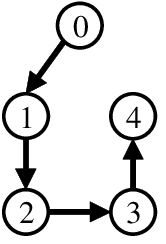}
			\caption{Network}
			\label{fig:ex_graph_2015}
		\end{minipage}
		\begin{minipage}{0.73\hsize}
			\centering
				\includegraphics[scale=1.0]{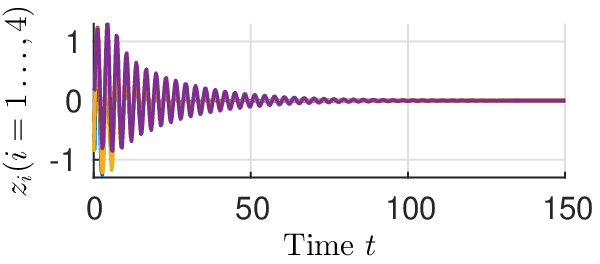}
		\caption{Output trajectories without uncertainty}
		\label{fig:ex_pre_2015_OK}
		\end{minipage}
	\end{figure}
	\begin{figure}[t]
		\centering
		\includegraphics[scale=1.0]{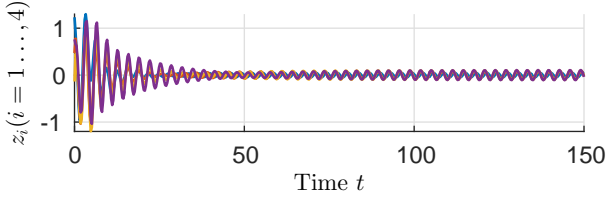}
		\caption{Output trajectories with uncertainty}
		\label{fig:ex_pre_2015}
	\end{figure}
	
	Fig.~\ref{fig:ex_pre_2015_OK} shows the simulation result without uncertainty parts $\Delta A_i$ and $\Delta B_i$ using the solution in \cite{cai2015cooperative} and Fig.~\ref{fig:ex_pre_2015} shows that with uncertainty.
	Observe in Fig.~\ref{fig:ex_pre_2015} that $z_i(i=i,\dots,4)$ do not converge to $0$ due to the uncertainty parts $\Delta A_i$ and $\Delta B_i$.
	Thus Problem~\ref{prob:mainproblem} cannot be solved by \cite{cai2015cooperative}, and we are motivated to propose a new solution that is not based on regulator equations but based on the internal model ($q_i$-copy).

	\section{STRUCTURE OF DISTRIBUTED CONTROLLER}
\label{sec:structuredist}
	At the outset we make the following (standard) assumptions.
	\begin{assumption}
		\label{ass:spanningtree}
		The digraph $\hat{\mathcal{G}}(t)$ uniformly contains a spanning tree and its root is node~0 (the exosystem).
	\end{assumption}
	
	\begin{assumption}
		\label{ass:stabilizable}
		For each agent $i\in\mathcal{V}$, $(A_{i0},B_{i0})$ is stabilizable.
	\end{assumption}
	
	\begin{assumption}
		\label{ass:detectable}
		For each agent $i\in\mathcal{V}$, $(C_{i0},A_{i0})$ is detectable.
	\end{assumption}
	
	\begin{assumption}
		\label{ass:transmissionzeros}
		For each agent $i\in\mathcal{V}$ and for every eigenvalues $\lambda$ of $S_0$,
			\begin{align}
			{\rm rank}\left[\begin{array}{cc}
				A_{i0}-\lambda I_{n_i} & B_{i0} \\
				C_{i0} & D_{i0}
			\end{array}\right]=n_i+q_i.\label{eq:transcondition}
			\end{align}
	\end{assumption}
	
	\begin{assumption}
		\label{ass:Seigen}
		The real parts of all eigenvalues of $S_0$ are zeros.
	\end{assumption}
	
	\begin{remark}		
		Assumption~\ref{ass:spanningtree} and Assumptions~\ref{ass:stabilizable}-\ref{ass:transmissionzeros} are necessary conditions
		for consensus over time-varying networks \cite{RenBeard:08} and for output regulation \cite{FraWon:76}, respectively. 
		Only Assumption~\ref{ass:Seigen} is a sufficient condition for (centralized) output regulation, but is commonly made for distributed output regulation (e.g. \cite{LiuHuang:15,cai2015cooperative}) such that the exogeneous signal does not diverge exponentially fast.
	\end{remark}
	\begin{remark}
		By \cite{Dav:76}, Assumption~\ref{ass:transmissionzeros} means that the {\it transmission zeros} of agent $i$ are disjoint from all eigenvalues of $S_0$, and
		also implies that the number of outputs is no more than that of inputs, i.e. $m_i\geq q_i$. A transmission zero $\zeta\in\mathbb{C}$ of agent $i$ is such that
		\begin{align*}
			{\rm rank}\left[\begin{array}{cc}
				A_{i0}-\zeta I_{n_i} & B_{i0} \\
				C_{i0} & D_{i0}
			\end{array}\right]<n_i+q_i.
		\end{align*}
	\end{remark}

	Because not all agents can access the exosystem (i.e. $w_0$ cannot be measured by all agents), we cannot use
	(\ref{eq:outputtransision}) directly. Instead we consider the following (estimated) error vector
	\begin{align}
		e_i=C_ix_i+D_iu_i+Q_iw_i\in\mathbb{R}^{q_i}\label{eq:errortransision}
	\end{align}
	where $w_i\in\mathbb{R}^r$ is the {\it estimated} exogeneous signal. This $e_i$ is $z_i$ in (\ref{eq:outputtransision}) with $w_0$ replaced by $w_i$.
	
	In order to solve Problem~\ref{prob:mainproblem}, we present a controller that consists of two parts: (1) distributed {\it exosystem generator} and
	(2) distributed {\it dynamic compensator}.

\subsection{Distributed exosystem generator}\label{subsec:deg}
	It is reasonable for each agent $i\in\mathcal{V}$ to have a local estimate of the exosystem's dynamics since not all agents can access the exosystem.
	Let $S_i(t)\in\mathbb{R}^{r\times r}$ be the estimete of $S_0$ and consider
	\begin{align}
		&\dot{S}_i(t)=\sum_{j=0}^N a_{ij}(t)\left(S_j(t)-S_i(t)\right),\label{eq:Strans}\\
		&\dot{w}_i(t)=S_i(t)w_i(t)+\sum_{j=0}^N a_{ij}(t)\left(w_j(t)-w_i(t)\right).\label{eq:exoest}
	\end{align}
	By using (\ref{eq:Strans}) and (\ref{eq:exoest}), it is guaranteed under Assumption~\ref{ass:spanningtree} that
	\begin{align*}
		\lim_{t\to\infty} (S_i(t)-S_0)=0,\quad\lim_{t\to\infty} (w_i(t)-w_0(t))=0
	\end{align*}
	for all $S_i(0)$, $w_i(0)$ and all $i \in \mathcal{V}$.
	We show this statement in detail in Section \ref{sec:main} below.
	
	This protocol is used to approximate the exosystem for each agent $i\in\mathcal{V}$.
	Thus we call (\ref{eq:Strans}) and (\ref{eq:exoest}) the {\it exosystem generator}.
	
	Equations (\ref{eq:Strans}) and (\ref{eq:exoest}) have also been used in \cite{LiuHuang:16} for the adaptive distributed observer (see Footnote 1 above), and first proposed in \cite{cai2015cooperative} but for time-invariant networks.
	
\subsection{Distributed dynamic compensator}
	We consider the following {\it dynamic compensator}
	\begin{align}
		&\dot{\xi}_i=E_i(t)\xi_i+F_i(t)e_i\nonumber\\
		&u_i=K_i(t)\xi_i\label{eq:compensator}
	\end{align}
	where $\xi_i$ is the state of the dynamic compensator and $e_i$ is defined in (\ref{eq:errortransision}).
	
	In order to specify the matrices $E_i(t),F_i(t),K_i(t)$ in (\ref{eq:compensator}), we extend the internal model control design in \cite[Section~1.3]{huang2004nonlinear} to the multi-agent system setting.
	Let $\lambda_{0,1},\dots,\lambda_{0,k}, k\leq r$ be the roots of the minimal polynomial of $S_0$. Note that $\{\lambda_{0,1},\dots,\lambda_{0,k}\}\subseteq\sigma(S_0)$.
	Then we define $\lambda_0\coloneqq[\lambda_{0,1}\cdots\lambda_{0,k}]^\top$.
	Let $c_{0,d}(\lambda_0),d=1,\dots,k$ be the coefficients of the polynomial satisfying
	\begin{align}
		s^k+c_{0,1}(\lambda_0)s^{k-1}+\cdots+c_{0,k-1}&(\lambda_0)s+c_{0,k}(\lambda_0)\nonumber\\
		&=\prod_{d=1}^k (s-\lambda_{0,d}(t)).\label{eq:minpoly}
	\end{align}
	
	For each agent $i\in\mathcal{V}$, let $\lambda_i(t)\coloneqq[\lambda_{i,1}(t)\cdots\lambda_{i,k}(t)]^\top$ be a local estimate of $\lambda_0$, and $c_{i,d}(\lambda_i),d=1,\dots,k$, 		
	the estimated coefficients generated by $\lambda_i(t)$ that satisfy
	\begin{align}
		s^k+c_{i,1}(\lambda_i)s^{k-1}+\cdots+c_{i,k-1}&(\lambda_i)s+c_{i,k}(\lambda_i)\nonumber\\
		&=\prod_{d=1}^k (s-\lambda_{i,d}(t)).\label{eq:minpoly_Si}
	\end{align}

	Consider the following consensus algorithm:
	\begin{align}
		\dot{\lambda}_i(t)=\sum_{j=0}^N a_{ij}(t)\left(\lambda_j(t)-\lambda_i(t)\right),\ \lambda_i(0)\in{\rm j}\mathbb{R}^k.
		\label{eq:lambdatrans}
	\end{align}
	It follows from Assumption~\ref{ass:spanningtree} that $\lambda_i(t)\rightarrow\lambda_0$ as $t\rightarrow\infty$. As a result, the coefficient $c_{i,d}(\lambda_i)\rightarrow c_{0,d}(\lambda_0)$ as $t\rightarrow\infty$
	for each $d=1,\dots,k$. Note that by Assumption~\ref{ass:Seigen} the entries of $\lambda_0$ are purely imaginary, and hence we only need to consider the initial condition
	$\lambda_i(0)\in{\rm j}\mathbb{R}^k$ (thus $\lambda_i(t)\in{\rm j}\mathbb{R}^k$ for all $t\geq0$).
	
	Since we consider that the agents' dynamics have uncertainty, the regulator equation approach (e.g. \cite{cai2015cooperative}) does not work.
	Thus for the robust output regulation problem, we consider the $q_i$-copy {\it internal model} as \cite[Section~1.3]{huang2004nonlinear}.
	In the case where an agent has multiple outputs, we need to assign the internal model to each output.
	Let $G_i(\lambda_i)\coloneqq I_{q_i}\otimes G'_i(\lambda_i),\ H_i\coloneqq I_{q_i}\otimes H'_i$ be the $q_i$-copy internal model ($\otimes$ denotes Kronecker product), where
	\begin{align}
		&G'_i(\lambda_i)\coloneqq\left[\begin{array}{cccc}
				0&1&\cdots&0 \\
				0&0&\cdots&0 \\
				\vdots&\vdots&\ddots&\vdots \\
				0&0&\cdots&1 \\
				-c_{i,k}(\lambda_i)&-c_{i,k-1}(\lambda_i)&\cdots&-c_{i,1}(\lambda_i)
			\end{array}\right]\nonumber\\
		&H'_i\coloneqq\left[\begin{array}{c}
				0 \\
				0 \\
				\vdots \\
				0 \\
				1
			\end{array}\right].\label{eq:Gprime}
	\end{align}
	We state the following lemma using the above matrices.
	\begin{lemma}
	\label{lem:stabilizableGH}
		Assume Assumption~\ref{ass:stabilizable} holds. Let $t \geq 0$ and $\lambda_i(t)=[\lambda_{i,1}(t) \cdots \lambda_{i,k}(t)]^\top$. 
		If for every $d\in\{1,\dots,k\}$, 
		\begin{align}
			{\rm rank}\left[\begin{array}{cc}
				A_{i0}-\lambda_{i,d}(t) I_{n_i} & B_{i0} \\
				C_{i0} & D_{i0}
			\end{array}\right]=n_i+q_i\label{eq:propertyofestimetedlambda}
		\end{align}
		then the following pair of matrices is stabilizable:
		\begin{align*}
			\left(\left[\begin{array}{cc}
				A_{i0}&0 \\
				H_iC_{i0}&G_i(\lambda_i(t))
			\end{array}\right],\left[\begin{array}{c}
				B_{i0} \\
				H_iD_{i0}
			\end{array}\right]\right).
		\end{align*}
	\end{lemma}
	\begin{proof}
Fix an arbitrary time $t \geq 0$, and let
\begin{align*}
M(\eta) := \left[\begin{array}{ccc}
				A_{i0}-\eta I_{n_i} & 0 & B_{i0} \\
				H_i C_{i0} & G_i(\lambda_i(t))-\eta I_{n_i} & H_i D_{i0}
			\end{array}\right].
\end{align*}
By the PBH test, the pair 
\begin{align*}
			\left(\left[\begin{array}{cc}
				A_{i0}&0 \\
				H_iC_{i0}&G_i(\lambda_i(t))
			\end{array}\right],\left[\begin{array}{c}
				B_{i0} \\
				H_iD_{i0}
			\end{array}\right]\right)
\end{align*}
is stabilizable if and only if 
\begin{align*}
{\rm rank}\ M(\eta) &= n_i + kq_i, \quad \forall \eta \in \mathbb{C}_+ \\ &\mbox{ (the closed right half complex plane)}. 
\end{align*}
Since $(A_{i0},B_{i0})$ is stabilizable by Assumption~3, ${\rm rank}\ [A_{i0}-\eta I_{n_i} \ B_{i0}] = n_i$ for all $\eta \in \mathbb{C}_+$. Also, ${\rm det}\, (G_i(\lambda_i(t))-\eta I_{n_i}) \neq 0$ for all $\eta \notin \sigma(G_i(\lambda_i(t)))$. Thus
\begin{align}\label{eq1}
{\rm rank}\ M(\eta) = n_i + kq_i, \quad \forall \eta \notin \sigma(G_i(\lambda_i(t))) \mbox{ and }\forall \eta \in \mathbb{C}_+. \tag{$\star$}
\end{align}
Write $M(\eta) = M_1(\eta) M_2(\eta)$, where
\begin{align*}
M_1(\eta) &:= \left[\begin{array}{ccc}
				I_{n_i} & 0 & 0 \\
				0 & H_i & G_i(\lambda_i(t))-\eta I_{n_i}
			\end{array}\right], \\
M_2(\eta) &:= \left[\begin{array}{ccc}
				A_{i0}-\eta I_{n_i} & 0 & B_{i0} \\
				C_{i0} & 0 & D_{i0} \\
				0 & I_{kq_i} & 0
			\end{array}\right].		
\end{align*}
Since $(G'_i(\lambda_i(t)), H'_i)$ given in (13) is in the control canonical form, the pair $(G_i(\lambda_i(t)), H_i)$ is controllable. Hence ${\rm rank}\ M_1(\eta)=n_i + kq_i$ for all $\eta \in \mathbb{C}$. On the other hand, it follows from 
\begin{align*}
\sigma(G_i(\lambda_i(t))) = \{\underbrace{\lambda_{i,1}(t),\ldots,\lambda_{i,1}(t)}_{q_i}, \ldots, \underbrace{\lambda_{i,k}(t),\ldots,\lambda_{i,k}(t)}_{q_i}\}
\end{align*}
and the condition~(14) below
\begin{align*}
			{\rm rank}\left[\begin{array}{cc}
				A_{i0}-\lambda_{i,d}(t) I_{n_i} & B_{i0} \\
				C_{i0} & D_{i0}
			\end{array}\right]=n_i+q_i, \quad \forall d=1,\dots,k
\end{align*}
that ${\rm rank}\ M_2(\eta)=n_i + kq_i + q_i$ for all $\eta \in \sigma(G_i(\lambda_i(t)))$. Therefore by Sylvester's inequality (i.e. $\min\{{\rm rank}\ M_1(\eta), {\rm rank}\ M_2(\eta)\} \geq {\rm rank}\ M(\eta) \geq {\rm rank}\ M_1(\eta) + {\rm rank}\ M_2(\eta) - (n_i + kq_i + q_i)$), we have
\begin{align}\label{eq2}
n_i + kq_i &\geq {\rm rank}\ M(\eta) \notag\\
&\geq (n_i + kq_i)+(n_i + kq_i + q_i)-(n_i + kq_i + q_i) \notag\\
&=n_i + kq_i, \quad \forall \eta \in \sigma(G_i(\lambda_i(t))). \tag{$\star \star$}
\end{align}
Combining (\ref{eq1}) and (\ref{eq2}) yields 
\begin{align*}
{\rm rank}\ M(\eta) = n_i + kq_i, \quad \forall \eta \in \mathbb{C}_+. 
\end{align*}
This proves that the pair 
\begin{align*}
			\left(\left[\begin{array}{cc}
				A_{i0}&0 \\
				H_iC_{i0}&G_i(\lambda_i(t))
			\end{array}\right],\left[\begin{array}{c}
				B_{i0} \\
				H_iD_{i0}
			\end{array}\right]\right)
\end{align*}
is stabilizable.
	\end{proof}
	
	In Lemma~\ref{lem:stabilizableGH} the sufficient condition (\ref{eq:propertyofestimetedlambda}) means that every $\lambda_{i,d}$ does not correspond to
	transmission zeros of agent $i$.
	In (\ref{eq:propertyofestimetedlambda}), $\lambda_{i,d}$ is time-varying because it is updated according to (\ref{eq:lambdatrans}).
	Since $\lambda_{i,d}(t)\in{\rm j}\mathbb{R}$ for all $t$, if agent $i$'s dynamics has purely imaginary transmission zeros,
	it is possible that (\ref{eq:propertyofestimetedlambda}) is violated. In order to satisfy (\ref{eq:propertyofestimetedlambda}) for all $t\geq0$, we make the following (simpifying) assumption.
	\begin{assumption}
		\label{ass:imaginaryaxis}
		For every agent $i\in\mathcal{V}$, there are no transmission zeros on the imaginary axis, i.e.
		\begin{align*}
			{\rm rank}\left[\begin{array}{cc}
				A_{i0}-\lambda I_{n_i} & B_{i0} \\
				C_{i0} & D_{i0}
			\end{array}\right]=n_i+q_i
		\end{align*}
		 for all $\lambda\in{\rm j}\mathbb{R}$.
	\end{assumption}
	
	If every agent $i\in\mathcal{V}$ is minimum-phase, then Assumption~\ref{ass:imaginaryaxis} is satisfied. In addition, this assumption allows transmission zeros to be on the
	open right (complex) plane, thus admitting non-minimum-phase system.
	In the case where Assumption~\ref{ass:imaginaryaxis} does not hold, it is a challenge to ensure that (\ref{eq:propertyofestimetedlambda}) holds for all $t\geq0$.
	Nevertheless, in Section~\ref{sec:moregeneral} below we shall present a novel strategy to guarantee (\ref{eq:propertyofestimetedlambda}) even in the presence of purely imaginary transmission zeros.

	From Lemma~\ref{lem:stabilizableGH} and Assumption~\ref{ass:imaginaryaxis}, we may synthesize $[K_{i1}(\lambda_i)\ K_{i2}(\lambda_i)]$ such that the matrix
	\begin{align}
		\left[\begin{array}{cc}
				A_{i0}&0 \\
				H_iC_{i0}&G_i(\lambda_i)
			\end{array}\right]+\left[\begin{array}{c}
				B_{i0} \\
				H_iD_{i0}
			\end{array}\right][K_{i1}(\lambda_i)\ K_{i2}(\lambda_i)]
			\label{eq:stablematrix}
	\end{align}
	is stable for all $t\geq0$. In addition, we choose $L_i$ such that the matrix $A_{i0}-L_iC_{i0}$ is stable under Assumption~\ref{ass:detectable}.
	
	Now we are ready to present the matrices $E_i(t),F_i(t)$ and $K_i(t)$ in the dynamic compensator (\ref{eq:compensator}):
	\begin{align}
		&E_i(\lambda_i)\coloneqq\nonumber\\
			&\ \left[\begin{array}{cc}
				A_{i0}-L_iC_{i0} & 0 \\
				0 & G_i(\lambda_i)
			\end{array}\right]+\left[\begin{array}{c} B_{i0}-L_iD_{i0} \\ 0 \end{array}\right]K_i(\lambda_i),\nonumber\\
		&F_i\coloneqq\left[\begin{array}{c}
				L_i \\
				H_i
			\end{array}\right],\nonumber\\
		&K_i(\lambda_i)\coloneqq[K_{i1}(\lambda_i)\ K_{i2}(\lambda_i)].\label{eq:EFandK}
	\end{align}

	Note that in (\ref{eq:EFandK}), $E_i$ and $K_i$ are time-varying as $\lambda_i$ is time-varying, while $F_i$ is time-invariant; and by (\ref{eq:lambdatrans}) there hold
	\begin{align*}
		&G_i(\lambda_i)\rightarrow G_i(\lambda_0)\\
		&K_i(\lambda_i)\rightarrow K_i(\lambda_0)\\
		&E_i(\lambda_i)\rightarrow E_i(\lambda_0).
	\end{align*}
	
	Using the distributed dynamic compesator (\ref{eq:compensator}), we will show in the next section that the estimated error $e_i$ and the output $z_i$ to be regulated converge to 0.

	\section{MAIN RESULT}
\label{sec:main}
	Our main result is the following.
	\begin{theorem}
	\label{thm:main}
		Consider the multi-agent system (\ref{eq:statetransision}), (\ref{eq:outputtransision}), (\ref{eq:perturbation}) and the exosystem (\ref{eq:exosystem}),
		and suppose that Assumptions~\ref{ass:apriori}-\ref{ass:imaginaryaxis} hold.
		Then for each agent $i\in\mathcal{V}$, the distributed exosystem generator (\ref{eq:Strans}) and (\ref{eq:exoest}), and the distributed dynamic compensator (\ref{eq:compensator})
		with (\ref{eq:lambdatrans}), (\ref{eq:EFandK}) solve Problem~\ref{prob:mainproblem}.
	\end{theorem}
	\medskip

	Several remarks on Theorem~\ref{thm:main} are in order.
	\begin{remark}
		Theorem~\ref{thm:main} asserts that the proposed two-part distributed controller -- the distributed exosystem generator (\ref{eq:Strans}), (\ref{eq:exoest}) and the distributed dynamic compensator (\ref{eq:compensator}) -- provides the first solution to the multi-agent output regulation problem where the agents have no initial knowledge of the exosystem's internal model and the agent dynamics are uncertain. The key of our solution is the time-varying $q_i$-copy internal model, updated locally based only on information received from neighbors, which eventually converges to the exact internal model of the exosystem.
	\end{remark}	
	\begin{remark}
		When there is only one agent (i.e. $N=1$), the problem is specialized to the centralized output regulation, and Theorem~\ref{thm:main} is thus an extension of the conventional results in \cite{FraWon:76,Fra:77,knobloch2012topics,huang2004nonlinear}. Even if the (single) agent does not know the exosystem dynamics initially, the output regulation problem is solvable by the exosystem generator (\ref{eq:Strans}), (\ref{eq:exoest}) and dynamic compensator (\ref{eq:compensator}).
	\end{remark}	
	\begin{remark}	
		If the exosystem is a leader agent that possesses computation and communication abilities, then the leader can compute the roots of its own minimal
		polynomial and send the information to other connected agents. If the exosystem is some entity that cannot compute or communicate, then those agents that can measure
		the exosystem (in particular know $S_0$) compute the corresponding minimal polynomial and the roots, and send the information to the rest of the network.
	\end{remark}
	\begin{remark}
		For each agent to `learn' the internal model of the exosystem, our strategy is to make the agents reach consensus by (\ref{eq:lambdatrans})
		for the {\it roots} of the exosystem's minimal polynomial (i.e. eigenvalues of $S_0$). It might appear more straightforward to reach consensus for the {\it coefficients}
		of the exosystem's minimal polynomial; the advantage of updating $\lambda_i$ with (\ref{eq:lambdatrans}), nevertheless, is that we may directly guarantee
		the equality in (\ref{eq:transcondition}) in Assumption~\ref{ass:transmissionzeros}.
	\end{remark}
	\begin{remark}
		\label{remark:S0_is_companion}
		In (\ref{eq:Strans}), there are $r\times r$ entries to update and communicate. If the minimal polynomial of $S_0$ equals its characteristic polynomial ($k=r$) and $S_0$ is in the companion form
		\begin{align*}
			S_0=\left[\begin{array}{cccc}
				0&1&\cdots&0 \\
				0&0&\cdots&0 \\
				\vdots&\vdots&\ddots&\vdots \\
				0&0&\cdots&1 \\
				-c_{0,r}(\lambda_0)&-c_{0,r-1}(\lambda_0)&\cdots&-c_{0,1}(\lambda_0)
			\end{array}\right]
		\end{align*}
		where $c_{0,d}(\lambda_0)$, $d=1,\dots,r$, and $\lambda_0$ are as defined in (\ref{eq:minpoly}), then each agent does not need to exchange and update the whole $S_i$.
		Each agent only needs to exchange and update $\lambda_i = [\lambda_{i,1},\dots,\lambda_{i,k}]^\top$ by (\ref{eq:lambdatrans}) and make $S_i$ also in the companion form.
	\end{remark}
	\begin{remark}
		\label{remark:reduce}
		In the equation (\ref{eq:lambdatrans}), we do not need to use all entries of $\lambda_i = [\lambda_{i,1},\dots,\lambda_{i,k}]^\top$, because the eigenvalues of the real matrices $S_i$ must be in conjugate pairs.
		Indeed, for all $i\in\hat{\mathcal{V}}$ we may write $\lambda_i$ in the following form	
		\begin{align*}
			\lambda_i(t)=\left\{\begin{array}{ll}
				\left[\hat{\lambda}_i(t)^\top\ \hat{\lambda}_i^*(t)^\top\right]^\top, & k\ {\rm is\ an\ even\ number} \\
				\left[\hat{\lambda}_i(t)^\top\ \hat{\lambda}_i^*(t)^\top\ 0\right]^\top, & k\ {\rm is\ an\ odd\ number}
			\end{array}\right.
		\end{align*}	
		where $\hat{\lambda}_i\in{\rm j}\mathbb{R}^{\lfloor k/2 \rfloor}$.	
		From this form, each agent can make their entire $\lambda_i$ after exchanging and updating only $\hat{\lambda}_i$.
	\end{remark}
	
	To prove Theorem~\ref{thm:main}, we need the following two lemmas. Their proofs are presented in Appendix (alternative proofs can also be found in \cite{LiuHuang:16}).
	
	The first lemma states a stability result for a particular type of time-varying systems.
	\begin{lemma}
		\label{lem:converge}
		Consider
		\begin{align}
			\dot{x}(t)=A_1(t)x(t)+A_2(t)x(t)+A_3(t)
			\label{eq:A1A2A3}
		\end{align}
		where $A_1(t)\in\mathbb{R}^{n\times n},A_2(t)\in\mathbb{R}^{n\times n},A_3(t)\in\mathbb{R}^{n}$ are piecewise continuous and bounded on $[0,\infty)$.
		Suppose that the origin is a uniformly exponentially stable equilibrium of $\dot{x}=A_1(t)x$, and $A_2(t)\rightarrow0$, $A_3(t)\rightarrow0$
		as $t\rightarrow\infty$.
		Then $x(t)\rightarrow0$ as $t\rightarrow\infty$.
	\end{lemma}
	
	The second lemma asserts that the distributed exosystem generators proposed in Section~\ref{subsec:deg} synchronize with the exosystem.
	\begin{lemma}
		\label{lem:Swconverge}
		Consider the distributed exosystem generator (\ref{eq:Strans}) and (\ref{eq:exoest}).
		If Assumption~\ref{ass:spanningtree} holds, then
		\begin{align*}
			\lim_{t\to\infty} (S_i(t)-S_0)=0,\quad\lim_{t\to\infty} (w_i(t)-w_0)=0
		\end{align*}
		for all $S_i(0), w_i(0)$.
	\end{lemma}
	
	Now we are ready to prove Theorem~\ref{thm:main}.
	
	{\noindent\hspace{2em}{\itshape Proof of Theorem~\ref{thm:main}: }}
		Let $\eta_i\coloneqq[x_i^\top\ \xi_i^\top]^\top$ be the combined state.
		From (\ref{eq:statetransision}), (\ref{eq:outputtransision}), (\ref{eq:errortransision}) and (\ref{eq:EFandK}), we derive
		\begin{align}
			\dot{\eta}_i &=M_i(\lambda_i)\eta_i+\left[\begin{array}{c}
					0 \\
					F_iQ_i
				\end{array}\right]w_i+\left[\begin{array}{c}
					P_i \\
					0
				\end{array}\right]w_0\label{eq:zetatranspr}\\
				z_i&=[C_i\ D_iK_i(\lambda_i)]\eta_i+Q_iw_0\label{eq:ztranspr}.
		\end{align}
		where
		\begin{align}
			&M_i(\lambda_i)\coloneqq \left[\begin{array}{cc} A_i & B_iK_i(\lambda_i)\\ F_iC_i & E_i(\lambda_i)+F_iD_iK_i(\lambda_i)\end{array} \right]\label{eq:matM}\\
					&\qquad\quad= M_{i0}(\lambda_i)+\Delta M_{i}(\lambda_i),\nonumber\\
			&M_{i0}(\lambda_i)=\nonumber\\
				&\scalebox{0.85}{$\left[\begin{array}{ccc}
					A_{i0} & B_{i0}K_{i1}(\lambda_i) & B_{i0}K_{i2}(\lambda_i) \\
					L_iC_{i0} & A_{i0}-L_iC_{i0}+B_{i0}K_{i1}(\lambda_i) & B_{i0}K_{i2}(\lambda_i)\\
					H_iC_{i0} & H_iD_{i0}K_{i1}(\lambda_i) & G_i(\lambda_i)+H_iD_{i0}K_{i2}(\lambda_i)
				\end{array}\right]$},\nonumber\\
			&\Delta M_{i}(\lambda_i)=\nonumber\\
				&\scalebox{0.75}{$\left[\begin{array}{ccc}
					\Delta A_i & \Delta B_iK_{i1}(\lambda_i) & \Delta B_iK_{i2}(\lambda_i)\\
					L_i\Delta C_i & L_i\Delta D_iK_{i1}(\lambda_i) & L_i\Delta D_iK_{i2}(\lambda_i)\\
					-\Delta A_i + L_i\Delta C_i & (L_i\Delta D_i - \Delta B_i)K_{i1}(\lambda_i) & (L_i\Delta D_i - \Delta B_i)K_{i2}(\lambda_i)
				\end{array}\right]	$}.\nonumber
		\end{align}
		First, we define
		\begin{align*}
		T \coloneqq \left[\begin{array}{ccc}
					I & 0 & 0 \\
					0 & 0 & I \\
					-I & I & 0
				\end{array}\right]
		\end{align*}
		and obtain
		\begin{align*}
			&TM_{i0}(\lambda_i)T^{-1} = \nonumber\\
				&\scalebox{0.8}{$\left[\begin{array}{cc|c}
					A_{i0}+B_{i0}K_{i1}(\lambda_i)		& B_{i0}K_{i2}(\lambda_i)						& B_{i0}K_{i1}(\lambda_i)\\
					H_iC_{i0}+H_iD_{i0}K_{i1}(\lambda_i)	& G_i(\lambda_i)+H_iD_{i0}K_{i2}(\lambda_i)	& H_iD_{i0}K_{i1}(\lambda_i)\\\hline
					0										& 0												& A_{i0}-L_iC_{i0}
				\end{array}\right]$}.
		\end{align*}
		Its upper-left submatrix equals (\ref{eq:stablematrix}) and thus the submatrix is stable.
		Also $A_i-L_iC_i$ is stable. Moreover $M_{i0}(\lambda_i)$ and $TM_{i0}(\lambda_i)T^{-1}$ are similar.
		Therefore $M_{i0}(\lambda_i)$ is stable for all $t\geq0$.
		By the continuity of eigenvalues, if the term $\Delta M_i(\lambda_i)$ is sufficiently small, then $M_i(\lambda_i)$ remains stable for all $t\geq0$. Indeed, as long as the uncertainty parts $\Delta A_i,\Delta B_i,\Delta C_i,\Delta D_i,\Delta P_i,\Delta Q_i$ in (\ref{eq:perturbation}) are such that the term $\Delta M_i(\lambda_i)$ does not perturb the stable eigenvalues of $M_{i0}(\lambda_i)$ to the closed right-hand-side of the complex plane, $M_i(\lambda_i)$ remains stable for all $t\geq0$.

		Next, from Assumption~\ref{ass:Seigen} and the above statement $\sigma(S_0)\cap\sigma(M_i(\lambda_0))=\emptyset$, and thus the following equations
		\begin{align}
			X_i(\lambda_0)S_0&=M_i(\lambda_0)X_i(\lambda_0)+\left[\begin{array}{c}
					P_i \\
					F_iQ_i
				\end{array}\right]\nonumber\\
				0&=[C_i\ D_iK_i(\lambda_0)]X_i(\lambda_0)+Q_i\label{eq:sylvester}
		\end{align}
		have a unique solution $X_i(\lambda_0)$ from \cite[Appendix~A]{knobloch2012topics}.
		Let $\tilde{\eta}_i\coloneqq\eta_i-X_i(\lambda_0)w_0$, $\tilde{M}_i\coloneqq M_i(\lambda_i)-M_i(\lambda_0)$.
		Then from (\ref{eq:zetatranspr}) and (\ref{eq:sylvester}), we obtain
		\begin{align*}
				\scalebox{0.9}{$\dot{\tilde{\eta}}_i$}&\scalebox{0.9}{$=\dot{\eta}_i-X_i(\lambda_0)\dot{w}_0$}\nonumber\\
				&\scalebox{0.9}{$=M_i(\lambda_i)\eta_i+\left[\begin{array}{c}
					0 \\
					F_iQ_i
				\end{array}\right]w_i+\left[\begin{array}{c}
					P_i \\
					0
				\end{array}\right]w_0-X_i(\lambda_0)S_0w_0$}\nonumber\\
				&\scalebox{0.9}{$=\left(\tilde{M}_i+M_i(\lambda_0)\right)\left(\tilde{\eta}_i+X_i(\lambda_0)w_0\right)+\left[\begin{array}{c}
					0 \\
					F_iQ_i
				\end{array}\right]w_i$}\nonumber\\
				&\quad\scalebox{0.9}{$+\left[\begin{array}{c}
					P_i \\
					0
				\end{array}\right]w_0-X_i(\lambda_0)S_0w_0$}\nonumber\\
				&\scalebox{0.9}{$=M_i(\lambda_0)\tilde{\eta}_i+\tilde{M}_i\tilde{\eta}_i+\tilde{M}_iX_i(\lambda_0)w_0+\left[\begin{array}{c}
					0 \\
					F_iQ_i
				\end{array}\right](w_i-w_0)$}\nonumber\\
				&\quad\scalebox{0.9}{$+\left(M_i(\lambda_0)X_i(\lambda_0)+\left[\begin{array}{c}
					P_i \\
					F_iQ_i
				\end{array}\right]-X_i(\lambda_0)S_0\right)w_0$}\nonumber\\
				&=\scalebox{0.9}{$\left(M_i(\lambda_0)\tilde{\eta}_i\right)+\left(\tilde{M}_i\tilde{\eta}_i\right)$}\nonumber\\
				&\quad\scalebox{0.9}{$+\left(\tilde{M}_iX_i(\lambda_0)w_0+\left[\begin{array}{c}
					0 \\
					F_iQ_i
				\end{array}\right](w_i-w_0)\right)$}.
		\end{align*}
		From Lemma~\ref{lem:Swconverge}, $w_i-w_0\rightarrow0$. Moreover, $M_i(\lambda_0)$ is stable and $\tilde{M}_i\rightarrow0$.
		Therefore, $\tilde{\eta}_i\rightarrow0$ from Lemma~\ref{lem:converge}.
		
		Furthermore from (\ref{eq:ztranspr}) and (\ref{eq:sylvester}), we obtain
		\begin{align*}
			z_i&=[C_i\ D_iK_i(\lambda_i)]\left(\tilde{\eta}_i+X_i(\lambda_0)w_0\right)+Q_iw_0\nonumber\\
			&=[C_i\ D_iK_i(\lambda_i)]\tilde{\eta}_i+\left([C_i\ D_iK_i(\lambda_i)]X_i(\lambda_0)+Q_i\right)w_0.
		\end{align*}
		Since $\tilde{\eta}_i\rightarrow0$ and
		\begin{align*}
			&[C_i\ D_iK_i(\lambda_i)]X_i(\lambda_0)+Q_i\nonumber\\
			&\quad\rightarrow[C_i\ D_iK_i(\lambda_0)]X_i(\lambda_0)+Q_i=0,
		\end{align*} 
		we conclude that $z_i(t)\rightarrow0$ as $t\rightarrow\infty$.
	{\hspace*{\fill}~\QED\par\endtrivlist\unskip}

	\begin{remark}\label{rem:perturbation}
		Note from the proof above that the uncertainty parts $\Delta A_i,\Delta B_i,\Delta C_i,\Delta D_i,\Delta P_i,\Delta Q_i$ need not be small.
		In particular, the matrices $\Delta P_i,\Delta Q_i$ can be arbitrary, and the matrices $\Delta A_i,\Delta B_i,\Delta C_i,\Delta D_i$ have only to satisfy that the matrix $M_i(\lambda_i)$ in (\ref{eq:matM}) is stable for all $t\geq0$.
	\end{remark}
	This requirement on the uncertainty parts is standard for centralized output regulation as well \cite{Dav:76}, and ideally should have been stated as a sufficient condition in Theorem~\ref{thm:main}. Inasmuch as this condition is rather clumsy (involving the matrix (\ref{eq:matM})) and for clarity of presentation, we choose to state it here as a remark. 

	\section{PURELY IMAGINARY TRANSMISSION ZEROS}
\label{sec:moregeneral}
	In this section, we generalize Theorem~\ref{thm:main} by designing a distributed controller for the case where Assumption~\ref{ass:imaginaryaxis} does not hold,
	i.e. there exist transmisssion zeros of the agents on the imaginary axis.
	In this case, because each (vector) $\lambda_i = [\lambda_{i,1},\dots,\lambda_{i,k}]^\top$ is updated continuously, entries of $\lambda_i$ may coincide with the transmission zeros of agent $i$, which would violate (\ref{eq:propertyofestimetedlambda}).
	Consequently we cannot design $K_i(\lambda_i)$ with Lemma~\ref{lem:stabilizableGH}.
	
	In order to choose the local estimate $\lambda_i$ satisfying the condition (\ref{eq:propertyofestimetedlambda}) and design $K_i(\lambda_i)$, 
	$\lambda_i$ must converge to $\lambda_0$ and at the same time avoid the transmission zeros of the agent $i$.
	Fig.~\ref{fig:routelambda} shows examples of the trajectory of $\lambda_{i,d}(d\in\{1,\dots,k\})$. The circles and the crosses represent
	respectively the transmission zeros of the agent $i$ and the eigenvalue $\lambda_{0,d}$ of $S_0$.
	\begin{figure}[t]
		\centering
		\includegraphics[scale=1.0]{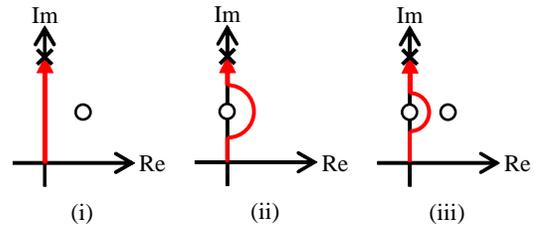}
		\caption{The trajectory of estimated eigenvalues $\lambda_{i,d}$; the circles represent the transmission zeros of the agent $i$, 
					the crosses represent the eigenvalues of $S_0$, and the arrows represent the trajectories of $\lambda_{i,d}(t).$}
		\label{fig:routelambda}
	\end{figure}
	The initial value $\lambda_{i,d}(0)$ is in ${\rm j}\mathbb{R}$ and $\lambda_{i,d}(t)$ moves toward $\lambda_{0,d}$.
	We divide the arrangement of transmission zeros into three cases:
	\begin{enumerate}
	\renewcommand{\labelenumi}{(\roman{enumi})}
		\item If there is no purely imaginary transmission zero of agent $i$ (see Fig.~\ref{fig:routelambda}(i)), then Assumption~\ref{ass:imaginaryaxis} holds and we need no further control design.
		\item If there is a purely imaginary transmission zero of agent $i$, and $\lambda_{i,d}(t)$ moves close to it (see Fig.~\ref{fig:routelambda}(ii)), then $\lambda_{i,d}(t)$ should move in a semicircle dented to the right around the transmission zero. By moving to the right, $\lambda_{i,d}(t)$ is always in $\mathbb{C}_+$ and thus Lemma~\ref{lem:stabilizableGH} is guaranteed for all $t\geq0$.
		\item If there is a purely imaginary transmission zero of agent $i$, and there are also other transmission zeros on the open right-half-plane (see Fig.~\ref{fig:routelambda}(iii)), the radius of semicircle should be smaller than (e.g. half of) the distance between these transmission zeros.
	\end{enumerate}	

	To formalize the above idea, we define several quantities.
	Let
	\begin{align}
		&\Pi_i\coloneqq\left\{s\in\mathbb{C}_+\ \middle|\ {\rm rank}\left[\begin{array}{cc}
			A_{i0}-sI_{n_i} & B_{i0} \\
			C_{i0} & D_{i0}
		\end{array}\right]<n_i+q_i \right\}\label{eq:TrZeroP}
	\end{align}
	be the set of closed right-half-plane transmission zeros of agent $i(\in\mathcal{V})$ and
	\begin{align}
		\tilde{\Pi}_i\coloneqq\left\{s\in\Pi_i\ \middle|\ {\rm Re}(s)=0 \right\}\label{eq:TrZeroImagP}
	\end{align}
	the subset of purely imaginary transmission zeros. We do not need to avoid open left-half-plane transmission zeros because $\lambda_{i,d}(t)\in\mathbb{C}_+$.
	Note that Assumption~\ref{ass:transmissionzeros} and $\Pi_i\cap\sigma(S_0)=\emptyset$ are equivalent, and Assumption~\ref{ass:imaginaryaxis} holds if and only if $\tilde{\Pi}_i=\emptyset$.

	We define a new function. For two sets $C_1,C_2\subseteq\mathbb{C}$ of finite number of complex numbers, define the {\it distance} between $C_1$ and $C_2$ by
	\begin{align*}
		{\rm dist}(C_1, C_2)\coloneqq\min\left\{|c_1-c_2|\ \middle|\  c_1\in C_1,c_2\in C_2\right\}.
	\end{align*}	
	Then define the radius $\rho_i\geq0$ of the semicircle shown in Fig.~\ref{fig:routelambda} as
	\begin{align}
		&\rho_i\coloneqq\nonumber\\
		&\scalebox{0.98}{$\left\{\begin{array}{ll}
			0,&\tilde{\Pi}_i=\emptyset\\
			\frac{1}{2}{\rm dist}(\sigma(S_0), \tilde{\Pi}_i), &\Pi_i\setminus\tilde{\Pi}_i=\emptyset\\
			\frac{1}{2}\min\{{\rm dist}(\sigma(S_0), \tilde{\Pi}_i), {\rm dist}(\Pi_i\setminus\tilde{\Pi}_i, \tilde{\Pi}_i)\}, &{\rm otherwise}.
		\end{array}\right.$}\label{eq:radiusofsemicircle}
	\end{align}
	This $\rho_i$ has three cases:
	\begin{enumerate}
	\renewcommand{\labelenumi}{(\roman{enumi})}
		\item If there is no purely imaginary transmission zeros of agent $i$, i.e. $\tilde{\Pi}_i=\emptyset$, then Assumption~\ref{ass:imaginaryaxis} holds and simply the radius is zero.
		\item If there are only purely imaginary transmission zeros of agent $i$, i.e. $\Pi_i\setminus\tilde{\Pi}_i=\emptyset$, then
			the radius is half of the distance between $\sigma(S_0)$ and $\tilde{\Pi}_i$.
		\item If there are transmission zeros of agent $i$ on both purely imaginary axis and right-half-plane, then
			the radius is half of the smaller of the distance between $\sigma(S_0)$ and $\tilde{\Pi}_i$ and that between $\Pi_i\setminus\tilde{\Pi}_i$ and $\tilde{\Pi}_i$.
	\end{enumerate}
	In the definition of $\rho_i$, we consider the coefficient $1/2$ for simplicity, but we can choose any coefficient from $(0,1)$.
	By using this $\rho_i$, we ensure the radius of the semicircle for the three cases as illustrated in Fig.~\ref{fig:routelambda}.
	
	Then we consider for all $i\in\hat{\mathcal{V}}$,
	\begin{align}
		&\lambda_i(t)=\alpha_i(t)+{\rm j}\beta_i(t)\in\mathbb{C}^{k},\ \alpha_i(t),\beta_i(t)\in\mathbb{R}^k\label{eq:newlambdatrans}\\
		&\dot{\beta}_i(t)=\sum_{j=0}^N a_{ij}(t)\left(\beta_j(t)-\beta_i(t)\right)\label{eq:betatrans}\\
		&\alpha_{i,d}(t)=\left\{
			\begin{array}{ll}
				0,&\gamma_{i,d}(t)\geq\rho_i{\rm\ or\ }i=0\\
				\sqrt{\rho_i^2-\gamma_{i,d}^2(t)},&\gamma_{i,d}(t)<\rho_i{\rm\ and\ }i\neq0
			\end{array}\right.\label{eq:alphatrans}
	\end{align}
	where
	\begin{align*}
		&d=1,\dots,k\\
		&\alpha_i(t)=[\alpha_{i,1}(t)\ \cdots\ \alpha_{i,k}(t)]^\top\\
		&\beta_i(t)=[\beta_{i,1}(t)\ \cdots\ \beta_{i,k}(t)]^\top\\
		&\gamma_{i,d}(t)\coloneqq\left\{
			\begin{array}{ll}
				0,&\tilde{\Pi}_i=\emptyset\\
				{\rm dist}\left(\{{\rm j}\beta_{i,d}(t)\},\tilde{\Pi}_i\right),&{\rm otherwise}.
			\end{array}\right.
	\end{align*}
	Note that $\gamma_{i,d}$ means the distance between ${\rm j}\beta_i$ and its closest (purely imaginary) transmission zero, and $\alpha_{i,d}(t)\geq0$ for all $t\geq0$.
	Moreover from Assumption~\ref{ass:spanningtree}, $\dot{\beta}_0(t) = 0, \alpha_0(t) = 0$ for all $t\geq0$.
	Finally, we update $\lambda_i$ by (\ref{eq:betatrans}) and (\ref{eq:alphatrans}) instead of (\ref{eq:lambdatrans}). It follows immediately from these definitions the following result.
	\begin{lemma}
		\label{lem:avoidingtranszeros}
		Consider the equations (\ref{eq:newlambdatrans}), (\ref{eq:betatrans}), (\ref{eq:alphatrans}). If Assumption~\ref{ass:spanningtree} holds, then
		$\lambda_{i,d}(t),i\in\mathcal{V}$, $d=1,\dots,k$, converge to $\lambda_{0,d}$ while avoiding transmission zeros of the agent $i$.
	\end{lemma}
	
	Using the above method, we state the main result of this section.
	\begin{theorem}
	\label{thm:moregeneral}
		Consider the multi-agent system (\ref{eq:statetransision}), (\ref{eq:outputtransision}), (\ref{eq:perturbation}) and the exosystem (\ref{eq:exosystem}),
		and suppose that Assumptions~\ref{ass:apriori}-\ref{ass:Seigen} hold.
		Then for each agent $i\in\mathcal{V}$, the distributed exosystem generator (\ref{eq:Strans}), (\ref{eq:exoest}) and the distributed dynamic compensator (\ref{eq:compensator})
		with (\ref{eq:EFandK}), (\ref{eq:newlambdatrans}), (\ref{eq:betatrans}), (\ref{eq:alphatrans}) solve Problem~\ref{prob:mainproblem}. 
	\end{theorem}
	\begin{proof}
		If Assumption~\ref{ass:spanningtree} holds, by Lemma~\ref{lem:Swconverge} we obtain $\lim_{t\to\infty} (S_i(t)-S_0)=0$ and $\lim_{t\to\infty} (w_i(t)-w_0)=0$.
		From Lemmas~\ref{lem:stabilizableGH}~and~\ref{lem:avoidingtranszeros}, each $\lambda_i$ satisfies the condition (\ref{eq:propertyofestimetedlambda})
		and thus we can define $K_i(\lambda_i)$ for all $t\geq0$. The remainder of the proof is identical to the proof of Theorem~\ref{thm:main}.
	\end{proof}
	\medskip
	\begin{remark}
		\label{remark:reduce2}
		As in Remark~\ref{remark:reduce}, we do not need to use all entries of $\beta_i$ in the equation (\ref{eq:betatrans}).
		For all $i\in\hat{\mathcal{V}}$ write $\beta_i$ in the following form	
		\begin{align}
			\beta_i(t)=\left\{\begin{array}{ll}
				\left[\hat{\beta}_i(t)^\top\ -\hat{\beta}_i(t)^\top\right]^\top, & k\ {\rm is\ an\ even\ number}\\
				\left[\hat{\beta}_i(t)^\top\ -\hat{\beta}_i(t)^\top\ 0\right]^\top, & k\ {\rm is\ an\ odd\ number}
			\end{array}\right.\label{eq:halfbeta}
		\end{align}	
		where $\hat{\beta}_i\in{\rm j}\mathbb{R}^{\lfloor k/2 \rfloor}$.	
		From this form, each agent can make their entire $\beta_i$ after exchanging and updating only $\hat{\beta}_i$.
	\end{remark}
	\begin{remark}
		As in Remark~\ref{rem:perturbation}, the uncertainty parts $\Delta A_i,\Delta B_i,\Delta C_i,\Delta D_i,\Delta P_i,\Delta Q_i$ need not be small.
	\end{remark}

	\section{OUTPUT SYNCHRONIZATION}
\label{sec:synchro}	
	In this section, we extend our approach to solve an {\it output synchronization} problem. Notable results of this problem are reported in \cite{scardovi2009synchronization, wieland2009internal, wieland2011internal, lunze2012synchronization, grip2015synchronization, seyboth2015robust, Zuo2017}. In \cite{scardovi2009synchronization}, the synchronization problem is solved for homogeneous agents. The result of \cite{scardovi2009synchronization} is extended by \cite{wieland2009internal, wieland2011internal} to deal with heterogeneous agents, and further extensions are reported in \cite{lunze2012synchronization, grip2015synchronization, seyboth2015robust, Zuo2017}. However, the approaches in \cite{scardovi2009synchronization, wieland2009internal, wieland2011internal, lunze2012synchronization, grip2015synchronization, seyboth2015robust, Zuo2017} cannot deal with uncertain agent dynamics and moreover assume certain common information available to all agents.	
	We address these issues by extending the approach developed in the previous sections.

	The output synchronization problem differs from the output regulation problem studied previously in that, there is no exosystem (node 0).	
	Consider $N$ heterogeneous agents whose dynamics are given by
	\begin{align}
		&\dot{x}_i=A_ix_i+B_iu_i \label{eq:statetransision_synch}\\
		&y_i=C_ix_i+D_iu_i \label{eq:outputtransision_synch}
	\end{align}
	where $x_i\in\mathbb{R}^{n_i}$ is the state vector, $u_i\in\mathbb{R}^{m_i}$ the control input, $y_i\in\mathbb{R}^{q}$ the output for $i=1,\dots,N$.
	Matrices $A_i,B_i,C_i,D_i$ may have different dimensions and entries and also have uncertainty as in the distributed output regulation problem, namely
	\begin{align}
		&A_i = A_{i0}+\Delta A_i,\ B_i = B_{i0}+\Delta B_i,\nonumber\\
		&C_i = C_{i0}+\Delta C_i,\ D_i = D_{i0}+\Delta D_i\label{eq:perturb_synch}
	\end{align}
	where $A_{i0},B_{i0},C_{i0},D_{i0}$ and $\Delta A_i,\Delta B_i,\Delta C_i,\Delta D_i$ are given in (\ref{eq:perturbation}).
	
	Since there does not exist the exosystem, we represent the time-varying interconnection among the $N$ agents by $\mathcal{G}(t)=(\mathcal{V},\mathcal{E}(t))$. We make the following assumption.

	\begin{assumption}
		\label{ass:underlyingdigraph}
		There is a fixed subset of nodes $\mathcal{V}_r\subseteq \mathcal{V}$ such that the digraph $\mathcal{G}(t)$ uniformly contains a spanning tree with respect to $\mathcal{V}_r$.
	\end{assumption}
	\begin{remark}
		In the distributed output regulation problem, it is necessary that the digraph $\hat{\mathcal{G}}(t)$ uniformly contains a single spanning tree whose root is the exosystem.
		By contrast, in the output synchronization problem, the digraph $\mathcal{G}(t)$ may uniformly contain multiple spanning trees with multiple roots. This is more general, although we require by Assumption~\ref{ass:underlyingdigraph} that these roots be time-invariant.
	\end{remark}	
	
	The output synchronization problem is the following;	
	\begin{problem}[Output Synchronization Problem]
	\label{prob:synchproblem}
		Given a network of agents (\ref{eq:statetransision_synch}), (\ref{eq:outputtransision_synch}) and (\ref{eq:perturb_synch}) with interconnection represented by $\mathcal{G}(t)$,
		design for each agent $i\in\mathcal{V}$ a distributed controller such that
		\begin{align*}
			\lim_{t\to\infty}\left(y_i(t)-y_j(t)\right)=0
		\end{align*}
		for all $i \neq j \in \mathcal{V}$ and all $x_i(0)$, $x_j(0)$.
	\end{problem}
	\medskip
	
	To solve Problem~\ref{prob:synchproblem}, we employ the same controller structure: (i) the distributed exosystem generator	\footnote{It is more appropriate to call this generator a ``distributed synchronizer generator'', as there is no exosystem in the current problem. However, since we use basically the same design as before, we simply inherit the same name.} and (ii) the distributed dynamic compensator.
	
\subsection{Distributed exosystem generator}
	First, we consider the distributed exosystem generator. We define $r\geq1$ as the dimension of the distributed exosystem generator as in Section~\ref{subsec:deg}. A natural choice for $r$ is $r=q$, the dimension of the output of each agent. The more general case where $r$ be different from $q$ can generate more diverse (interesting) synchronized patterns. For example, when $q=1$ and $r=2$, we can make the final synchronized patterns to be constant, ramp, or sinusoidal by choosing suitable second-order exosystem generators. An illustrating example is provided below in Section~\ref{EX:OS1}.
	
	To solve Problem~\ref{prob:synchproblem}, we consider again the distributed exosystem generator:
	\begin{align}
		&\dot{S}_i(t) = \sum_{j=1}^N a_{ij}(t)\left(S_j(t)-S_i(t)\right),\nonumber\\
		&\dot{w}_i(t)=S_i(t)w_i(t)+\sum_{j=1}^N a_{ij}(t)\left(w_j(t)-w_i(t)\right)\label{eq:DEG_synch}
	\end{align}
	where for each $i \in \mathcal{V}_r$, $S_i(0) = S^* \in \mathbb{R}^{r \times r}$ is a fixed matrix; for each $i \in \mathcal{V}\setminus\mathcal{V}_r$, $S_i(0) \in \mathbb{R}^{r \times r}$ is arbitrary; and for each $i \in \mathcal{V}$, $w_i(0) \in \mathbb{R}^r$ is arbitrary. For the fixed $S^*$ we need the following assumption similar to Assumption~\ref{ass:Seigen} (for $S_0$ in the distributed output regulation problem).
	\begin{assumption}
		\label{ass:Seigen_synch}
		The real parts of all eigenvalues of $S^*$ are zeros.
	\end{assumption}
	
	We require that the agents in $\mathcal{V}_r$ have the same initial dynamics $S^*$. This is to derive the following convergence result, as for the distributed exosystem generator (\ref{eq:Strans}) and (\ref{eq:exoest}). This requirement on the initial condition might be stringent, but it already relaxes the requirement in the literature \cite{scardovi2009synchronization, wieland2009internal, wieland2011internal, lunze2012synchronization, grip2015synchronization, seyboth2015robust, Zuo2017} where all agents must have the same dynamics for synchronization.
	
	\begin{lemma}
		\label{lem:Swconverge_synch}
		Consider the distributed exosystem generator (\ref{eq:DEG_synch}).
		If Assumption~\ref{ass:underlyingdigraph} holds, then
		\begin{align*}
			\lim_{t\to\infty} (S_i(t)-S^*)=0,\quad\lim_{t\to\infty} (w_i(t)-w_j(t))=0
		\end{align*}
		for $(\forall i \in \mathcal{V}_r)S_i(0)=S^*$, $(\forall i \in \mathcal{V}\setminus\mathcal{V}_r)S_i(0)$, and $(\forall i,j \in \mathcal{V})w_i(0), w_j(0)$.
	\end{lemma}
	The proof is in Appendix.
	
\subsection{Distributed dynamic compensator}
	\label{sec:DDC_synch}
	As in the output regulation problem, we solve the output synchronization problem by reducing the error between the output $y_i \in \mathbb{R}^q$ of each agent and the exogeneous signal $w_i\in\mathbb{R}^r$.
	
	When $q=r$, the error is simply $e_i = y_i-w_i$. Since we consider the more general case where $q$ need not be equal to $r$, we define the error to be 
	\begin{align}
		e_i = y_i + Q_i w_i.\label{eq:error_synch}
	\end{align}
	Here $Q_i \in \mathbb{R}^{q \times r}$ and may be different for different agents. 
	Thus for output synchronization, it is important that $Q_i$, $w_i$ converge to the same vector for all agents. Since $w_i$ do so by Lemma~\ref{lem:Swconverge_synch}, we propose the following consensus update for $Q_i$:
	\begin{align}
		\dot{Q}_i=\sum_{j=1}^N a_{ij}(t)\left(Q_j(t)-Q_i(t)\right),\ Q_i(0)\in\mathbb{R}^{q\times r}\label{eq:Q_consensus_synch}
	\end{align}
	for all $i,j\in\mathcal{V}$.
	\begin{lemma}
		\label{lem:reference_synch}
		Consider the distributed exosystem generator (\ref{eq:DEG_synch}) and (\ref{eq:Q_consensus_synch}).
		If Assumption~\ref{ass:underlyingdigraph} holds, then
		\begin{align*}
			\lim_{t\to\infty} (Q_i(t)w_i(t)-Q_j(t)w_j(t))=0
		\end{align*}
		for all $i,j\in\mathcal{V}$ and all $Q_i(0)$, $Q_j(0)$, $w_i(0)$, $w_j(0)$.
	\end{lemma}
	\begin{proof}
		From \cite[Theorem~1]{moreau2004stability} and (\ref{eq:Q_consensus_synch}), $Q_i(i\in\mathcal{V})$ reach consensus for all $Q_i(0)$. Let $Q^*$ be the consensus value of $Q_i(t)$.
		Then
		\begin{align*}
			Q_iw_i&-Q_jw_j\\
			&=(Q_i-Q^*)w_i+Q^*w_i-(Q_j-Q^*)w_j-Q^*w_j\\
			&=(Q_i-Q^*)w_i-(Q_j-Q^*)w_j+Q^*(w_i-w_j)
		\end{align*}
		Since $(w_i-w_j)\rightarrow0$ for all $w_i(0)$, $w_j(0)$ from Lemma~\ref{lem:Swconverge_synch} and $(Q_i-Q^*)\rightarrow0$, $(Q_j-Q^*)\rightarrow0$, we ensure $Q_i(t)w_i(t)-Q_j(t)w_j(t)\rightarrow 0$ as $t\rightarrow\infty$.		
	\end{proof}
	
	We consider again the dynamic compensator
	\begin{align}
		&\dot{\xi}_i=E_i(t)\xi_i+F_ie_i(t)\nonumber\\
		&u_i=K_i(t)\xi_i\label{eq:compensator_synch}
	\end{align}
	where $\xi_i$ is the state of the dynamic compensator and $e_i$ is defined in (\ref{eq:error_synch}).
	The matrices $E_i(t)$, $F_i$, and $K_i(t)$ are as specified in (\ref{eq:EFandK}), (\ref{eq:newlambdatrans}), (\ref{eq:betatrans}) and (\ref{eq:alphatrans}). As in Section~\ref{sec:moregeneral}, such a dynamic compensator can deal with purely imaginary transmission zeros. Note that in (\ref{eq:betatrans}), the initial values of $\beta_i$ are	
	\begin{align*}
		&(\forall i \in\mathcal{V}_r)\ \beta_i(0) = [{\rm Im}\{\lambda_1(S^*)\}\cdots{\rm Im}\{\lambda_r(S^*)\}]^\top\\
		&(\forall i \in\mathcal{V}\setminus\mathcal{V}_r)\ \beta_i(0) \in\mathbb{R}^{r}.
	\end{align*}
	
	In the next subsection, we present the result of the output synchronization problem.
	
\subsection{Result}
	Our result for the output synchronization problem is the following.
	\begin{theorem}
	\label{thm:synchro}
		Consider the multi-agent system (\ref{eq:statetransision_synch}), (\ref{eq:outputtransision_synch}), (\ref{eq:perturb_synch}), and suppose that Assumptions~\ref{ass:stabilizable}, \ref{ass:detectable}, \ref{ass:transmissionzeros}, \ref{ass:underlyingdigraph}, \ref{ass:Seigen_synch} hold.
		Then for each agent $i\in\mathcal{V}$, the distributed exosystem generator (\ref{eq:DEG_synch}), and the distributed dynamic compensator (\ref{eq:compensator_synch}) with (\ref{eq:EFandK}), (\ref{eq:newlambdatrans}), (\ref{eq:betatrans}), (\ref{eq:alphatrans}) and (\ref{eq:Q_consensus_synch}) solve Problem~\ref{prob:synchproblem}.
	\end{theorem}
	Theorem~\ref{thm:synchro} improves the results of \cite{wieland2011internal} in the following aspects. (i) While \cite{wieland2011internal} requires $S^*$ and $Q^*$ to be known as common information by all agents, we allow $Q_i (i \in \mathcal{V})$ and $S_i (i \in \mathcal{V}\setminus\mathcal{V}_r)$ to be different. This makes our solution more suitable for distributed implementation. Only in the special case when $\mathcal{V}_r = \mathcal{V}$ does our requirement on $S_i$ become the same as \cite{wieland2011internal}. (ii) While \cite{wieland2011internal} cannot deal with uncertain agent dynamics, we address uncertainty by the ($q$-copy) internal model principle.
	
	{\noindent\hspace{2em}{\itshape Proof of Theorem~\ref{thm:synchro}: }}
		From Lemma~\ref{lem:avoidingtranszeros} with (\ref{eq:newlambdatrans}), (\ref{eq:betatrans}), (\ref{eq:alphatrans}), for all $i\in\mathcal{V}$, $\lambda_i(t)$ achieve consensus while avoiding transmission zeros of agent $i$ and the consensus value is the vector of eigenvalues of $S^*$.
		From Lemma~\ref{lem:reference_synch}, we have $Q_iw_i\rightarrow w_{\rm ref}^*\in\mathbb{R}^r,i\in\mathcal{V}$ where $w_{\rm ref}^*$ is some constant vector.
		In the same way as in the proof of Theorem~\ref{thm:main} with $P_i=0$, we obtain $e_i\rightarrow0$ and
		\begin{align*}
			y_i = e_i - Q_iw_i\rightarrow -w_{\rm ref}^*, i\in\mathcal{V}.
		\end{align*}
		Therefore $(y_i(t)-y_j(t))\rightarrow0$.		
	{\hspace*{\fill}~\QED\par\endtrivlist\unskip}
	
	Note that as in Remark~\ref{rem:perturbation}, the uncertainty parts $\Delta A_i,\Delta B_i,\Delta C_i,\Delta D_i$ need not be small, and only need to satisfy that the matrix $M_i(\lambda_i)$ in (\ref{eq:matM}) is stable for all $t\geq0$.

	\section{SIMULATION EXAMPLES}
\label{sec:example}
	In this section, we illustrate the designed distributed controller by applying it to solve distributed output regulation problems, as well as an
	output synchronization problem.
\subsection{Distributed Output Regulation Problem}
\subsubsection{Example in Section \ref{sec:CompEx}, continued}
	\label{EX:OR1}
	Consider the time-varying network as displayed in Fig.~\ref{fig:topologyex}. The network periodically switches between $\hat{\mathcal{G}}_1$ and $\hat{\mathcal{G}}_2$ every 10 seconds.
	Thus this network uniformly contains a spanning tree and the root is node~0 (indeed the network in Fig.~\ref{fig:ex_graph_2015}). Therefore Assumption~\ref{ass:spanningtree} holds.	
	\begin{figure}[t]
		\centering
		\includegraphics[scale=1.0]{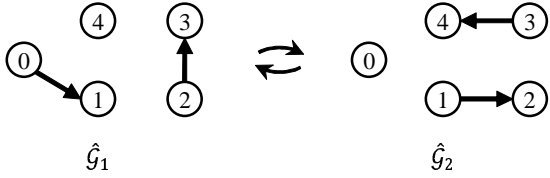}
		\caption{Time-varying network (Example \ref{EX:OR1})}
		\label{fig:topologyex}
	\end{figure}

	Also consider the exosystem and 4 agents with exactly the same parameters as Section \ref{sec:CompEx}.
	Then $\beta_0$ in (\ref{eq:betatrans}) is $\beta_0=[2\ -2]^\top$. It is checked that $(A_{i0},B_{i0})$ and $(C_i,A_{i0})$ are controllable and observable, respectively, and therefore stabilizable and detectable, respectively; thus {Assumptions~\ref{ass:stabilizable},\ref{ass:detectable} hold.
	Then we choose $L_i$ such that the eigenvalues of $A_{i0}-L_iC_i$ are $\{-1,-2,-3\}$.
	
	The transmission zeros of the agents are 
	\begin{align}
		\Pi_i=\tilde{\Pi}_i=\{\pm(0.5+0.1i){\rm j}\}\label{eq:exampletrz}
	\end{align}
	Therefore, all the transmission zeros are on the imaginary axis and Assumption~\ref{ass:imaginaryaxis} does not hold.

	We choose $w_0(0)$ uniformly at random from the interval $[-1, 1]$.
	We apply the distributed exosystem generator (\ref{eq:Strans}) and (\ref{eq:exoest}) with the initial conditions $w_i(0)$ selected uniformly
	at random from the interval $[-1,1]$, and set
	\begin{align*}
		S_i(0)=\left[\begin{array}{cc}0 & 0.5i\\ -0.5i & 0\end{array}\right].
	\end{align*}
	We also apply the distributed dynamic compensator (\ref{eq:compensator}), (\ref{eq:EFandK}), (\ref{eq:newlambdatrans}),
	(\ref{eq:betatrans}), (\ref{eq:alphatrans}) with initial conditions $x_i(0)$ and $\xi_i(0)$ selected uniformly at random from the interval $[-1,1]$, and set $\beta_i(0)=0$ for all $i\in\mathcal{V}$.
	From (\ref{eq:exampletrz}), $\rho_i$ and $\gamma_{i,d}(t)$ in (\ref{eq:alphatrans}) are 
	\begin{align*}
		&\rho_i =(2-(0.5+0.1i))/2\\
		&\gamma_{i,1}(t)=|0.5+0.1i-\beta_{i,1}(t)|\\
		&\gamma_{i,2}(t)=|-(0.5+0.1i)-\beta_{i,2}(t)|.
	\end{align*}
	
	\begin{figure}[t]
		\centering
		\includegraphics[scale=1.0]{REGULATEDRESULT_state.eps}
		\caption{Simulation result (distributed output regulation problem with 1 output)}
		\label{fig:result}
	\end{figure}
	\begin{figure}[!t]
		\centering
		\includegraphics[scale=1.0]{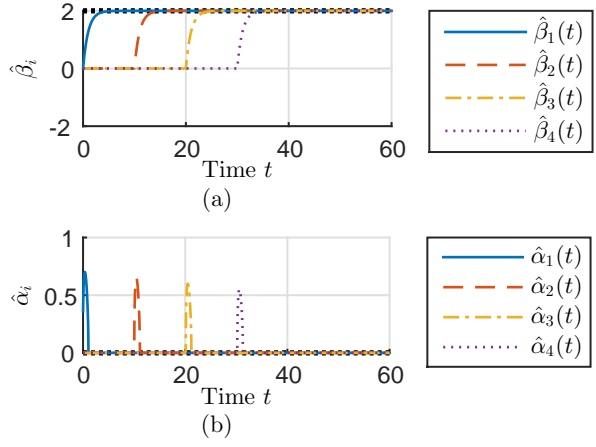}
		\caption{Trajectories of real and imaginary parts of local esitimete $\hat{\lambda}_i$ ($\hat{\beta}_0(t) = 2, \hat{\alpha}_0(t) = 0$)}
		\label{fig:resulteigen}
	\end{figure}
	The simulation result is displayed in Fig.~\ref{fig:result}.
	In Fig.~\ref{fig:result}(a), the dotted curve represents the first element of exosystem's signal $w_{0,1}$ and others represent the estimated exogeneous signals $w_{i,1}, i=1,2,3,4$.
	Observe that all $w_{i,1}$ synchronize with $w_{0,1}$. Thus the distributed exosystem generators effectively
	create a local copy of the exosystem, despite that not all agents have access to the exosystem and the network is time-varying.
	
	Fig.~\ref{fig:result}(b) shows the regulated output $z_i$ of each agent (in this simulation, $z_i=x_{i,1}-w_{0,1}$). Observe that all $z_i$ converge to 0.
	This demonstrates the effectiveness of the distributed dynamic compensators for achieving perfect regulation, despite of the parameter perturbation and initially imprecise internal model of the exosystem.
	
	We next examine the parameters in the distributed dynamic compensator. Define $\hat{\alpha}_i$ as the real part of local estimate $\lambda_i$ made in the same as (\ref{eq:alphatrans}) with $\hat{\beta}_i$ in (\ref{eq:halfbeta}).
	In this example, $\hat{\alpha}_i,\hat{\beta}_i\in\mathbb{R}$, $\alpha_i = [\hat{\alpha}_i\ \hat{\alpha}_i]^\top\in\mathbb{R}^2$ and $\beta_i = [\hat{\beta}_i\ -\hat{\beta}_i]^\top\in\mathbb{R}^2$.
	Fig.~\ref{fig:resulteigen}(a) and (b) show $\hat{\beta}_i$ and $\hat{\alpha}_i$, respectively. Each $\hat{\beta}_i$ converges to $\hat{\beta}_0$, and each $\hat{\alpha}_i$ becomes positive exactly
	when the distance between $\hat{\beta}_i$ and the closest transmission zeros to $\hat{\beta}_i$, namely $\gamma_{i,d}(t)$, is less than $\rho_i$.
		
	\begin{figure}[t]
		\centering
		\includegraphics[scale=1.0]{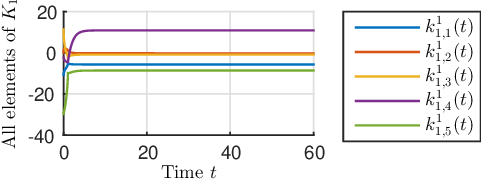}
		\caption{Trajectories of $K_1=[k_{1,1}^1\ k_{1,2}^1\ k_{1,3}^1\ k_{1,4}^1\ k_{1,5}^1]$}
		\label{fig:resultK1}
	\end{figure}
	\begin{figure}[!t]
		\centering
		\includegraphics[scale=1.0]{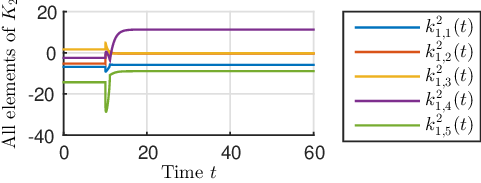}
		\caption{Trajectories of $K_2=[k_{1,1}^2\ k_{1,2}^2\ k_{1,3}^2\ k_{1,4}^2\ k_{1,5}^2]$}
		\label{fig:resultK2}
	\end{figure}
	\begin{figure}[t]
		\centering
		\includegraphics[scale=1.0]{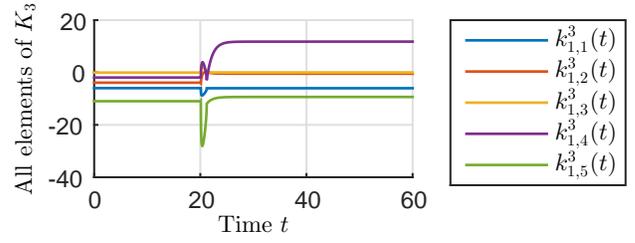}
		\caption{Trajectories of $K_3=[k_{1,1}^3\ k_{1,2}^3\ k_{1,3}^3\ k_{1,4}^3\ k_{1,5}^3]$}
		\label{fig:resultK3}
	\end{figure}
	\begin{figure}[!t]
		\centering
		\includegraphics[scale=1.0]{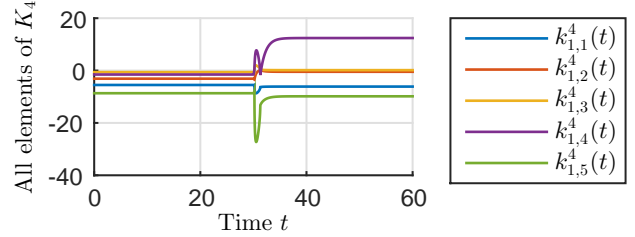}
		\caption{Trajectories of $K_4=[k_{1,1}^4\ k_{1,2}^4\ k_{1,3}^4\ k_{1,4}^4\ k_{1,5}^4]$}
		\label{fig:resultK4}
	\end{figure}
	The internal model $G_i(t)$'s entries contain $\alpha_i(t)$ and $\beta_i(t)$. In this example, $G_i(t)$ is in the form
	\begin{align*}
		G_i(t)=\left[\begin{array}{cc}0 & 1\\ -\left(\hat{\alpha}_i(t)^2+\hat{\beta}_i(t)^2\right) & 2\hat{\alpha}_i(t)\end{array}\right],
	\end{align*}
	and we choose the matrix $K_i(t)$ such that the eigenvalues of (\ref{eq:stablematrix}) are $\left\{-0.4, -0.8, -1.2, -1.6, -2.0\right\}$.

	Figs.~{\ref{fig:resultK1}~-~\ref{fig:resultK4}} show the trajectories of all elements of $K_i(t)=[k^i_{1,1}(t)\ \cdots\ k^i_{1,5}(t)],i=1,\dots,4$ in this example, respectively.
	Observe that each entry of $K_i$ changes significantly exactly when $\alpha_i$ becomes positive (i.e. avoiding transmission zeros).
	
\subsubsection{Output regulation of two-dimensional outputs over a large scale network}
	\label{EX:OR2}
	Consider the large scale time-varying network as displayed in Fig.~\ref{fig:treelarge}. The network periodically switches among $\hat{\mathcal{G}}_1$, $\hat{\mathcal{G}}_2$ and $\hat{\mathcal{G}}_3$ every 2, 3 and 5 seconds, respectively.
	Thus this network uniformly contains a spanning tree and the root is node~0. Therefore Assumption~\ref{ass:spanningtree} holds.
	
	The exosystem (node~0) is
	\begin{align*}
		\dot{w}_0(t)=S_0w_0,\quad S_0=\left[\begin{array}{cc}0 & 1\\ -1 & 0\end{array}\right].
	\end{align*}
	
	Excluding the exosystem the network contains 155 agents and they are classified into five types:
	\begin{align*}
		&\dot{x}_i=A_ix_i+B_iu_i+P_iw_0\\
		&z_i=C_ix_i+D_iu_i+Q_iw_0
	\end{align*}
	where
	\begin{align*}
		&A_i = A_{i0}+\Delta A,\ B_i = B_{i0}+\Delta B,\\
		&A_{i0}=\left[\begin{array}{cc}0 & 1\\ m_i & 2\end{array}\right],\ \Delta A = \left[\begin{array}{cc}0 & 0\\ 0 & -0.1\end{array}\right],\\
		&B_{i0}=\left[\begin{array}{cc}1 & 0\\ (0.1m_i+0.2)^2+m_i+1 & 1\end{array}\right],\\
		&\Delta B = \left[\begin{array}{cc}0 & 0\\ -0.1 & 0\end{array}\right],\ C_i=\left[\begin{array}{cc}1 & 0\\ 0 & 1\end{array}\right],\ D_i=\left[\begin{array}{cc}1 & 0\\ 0 & 1\end{array}\right],\\
		&P_i=\left[\begin{array}{cc}0 & 1\\ 1 & 0\end{array}\right],\ Q_i=\left[\begin{array}{cc}-1 & 0\\ 0 & -1\end{array}\right],
	\end{align*}
	and $m_i = 1$ for $i=1,6,11,\dots$, $m_i = 2$ for $i=2,7,12,\dots$, $m_i = 3$ for $i=3,8,13,\dots$, $m_i = 4$ for $i=4,9,14,\dots$ and $m_i = 5$ for $i=5,10,15,\dots$.
	It is checked that $(A_{i0},B_{i0})$ and $(C_i,A_{i0})$ are controllable and observable, respectively, and therefore stabilizable and detectable, respectively; thus Assumptions~\ref{ass:stabilizable}, \ref{ass:detectable} hold.
	Then we choose $L_i$ such that the eigenvalues of $A_{i0}-L_iC_i$ are $\{-1.20,-1.21\}$.	
	
	The transmission zeros of the agents are 
	\begin{align}
		\Pi_i=\tilde{\Pi}_i=\{\pm(0.1m_i+0.2){\rm j}\}
	\end{align}
	Therefore, all the transmission zeros are on the imaginary axis and Assumption~\ref{ass:imaginaryaxis} does not hold.	

	We choose $w_0(0)$ uniformly at random from the interval $[-1, 1]$.
	We apply the distributed exosystem generator (\ref{eq:Strans}) and (\ref{eq:exoest}) with the initial conditions $w_i(0)$ selected uniformly
	at random from the interval $[-1,1]$, and set
	\begin{align*}
		S_i(0)=0_{2\times2}.
	\end{align*}
	\begin{figure}[t]
		\centering
		\includegraphics[scale=1.0]{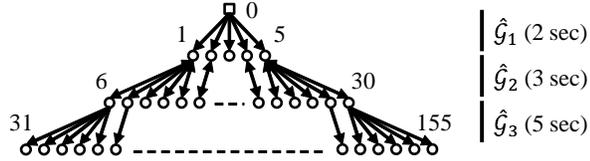}
		\caption{Time-varying network of $155$ agents (Example \ref{EX:OR2})}
		\label{fig:treelarge}
	\end{figure}
	\begin{figure}[t]
		\centering
		\includegraphics[scale=1.0]{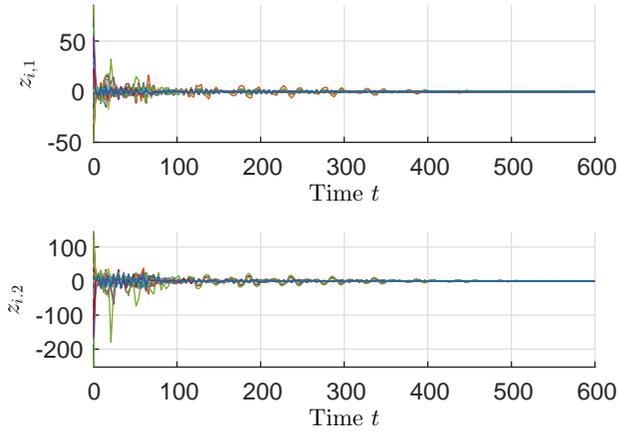}
		\caption{Simulation result (large scale network ($i=1,\dots,155$))}
		\label{fig:errorlarge}
	\end{figure}
	We also apply the distributed dynamic compensator (\ref{eq:compensator}), (\ref{eq:EFandK}), (\ref{eq:newlambdatrans}),
	(\ref{eq:betatrans}), (\ref{eq:alphatrans}) with initial conditions $x_i(0), \xi_i(0)$ selected uniformly at random from the interval $[-1,1]$ and $\beta_i(0) = 0$ for all $i\in\mathcal{V}$.
	We choose the matrix $K_i(t)$ such that the eigenvalues of (\ref{eq:stablematrix}) are $\left\{-0.70,-0.71,-0.72,-0.73,-0.74,-0.75\right\}$.

	The simulation result is displayed in Fig.~\ref{fig:errorlarge}. This figure shows the regulated output $z_{i,1}, z_{i,2}$ of each agent (in this simulation, $z_{i,1}=x_{i,1}-w_{0,1}$ and $z_{i,2}=x_{i,2}-w_{0,2}$).
	Observe that all $z_{i,1}$ and $z_{i,2}$ converge to 0 for all $i\in\mathcal{V}$. This demonstrates the effectiveness of $q_i$-copy internal model for robust regulation of higher dimensional outputs over large scale networks.

\subsection{Distributed Output Synchronization Problem}
	\label{EX:OS1}
	\begin{figure}[t]
		\centering
		\includegraphics[scale=1.0]{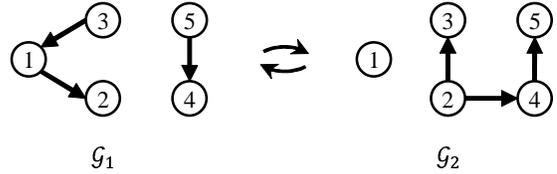}
		\caption{Time-varying network (Example \ref{EX:OS1})}
		\label{fig:topologyex2}
	\end{figure}
	Consider the time-varying network in Fig.~\ref{fig:topologyex2}.
	The network periodically switches between $\mathcal{G}_1$  and $\mathcal{G}_2$ every 3 seconds.
	Thus this network uniformly contains a spanning tree with respect to $\mathcal{V}_r =\{1,2,3\}$. 
	Therefore Assumption~\ref{ass:underlyingdigraph} holds.
	
	The agents ($i=1,\dots,5$) are
	\begin{align*}
		&\dot{x}_i=A_ix_i+B_iu_i\\
		&y_i=C_ix_i+D_iu_i
	\end{align*}
	where
	\begin{align*}
		&A_i = A_{i0}+\Delta A,\ B_i = B_{i0}+\Delta B,\\
		&A_{i0}=\left[\begin{array}{cc}0 & i\\ -i & 0 \end{array}\right],\
		B_{i0}=\left[0\ -0.99i\right]^\top,\\
		&\Delta A = \left[\begin{array}{cc}0.1 & 0\\ 0 & 0 \end{array}\right],\
		\Delta B = \left[\begin{array}{c}0.1 \\ 0\end{array}\right],\
		C_i=\left[1\ 0\right],\
		D_i=1.
	\end{align*}
	It is checked that $(A_{i0},B_{i0})$ and $(C_i,A_{i0})$ are controllable and observable, respectively, and therefore stabilizable and detectable, respectively; thus {Assumptions~\ref{ass:stabilizable},\ref{ass:detectable} hold.
	Then we choose $L_i$ such that the eigenvalues of $A_{i0}-L_iC_i$ are $\{-0.7,-0.8\}$.
	
	The transmission zeros of the agents are 
	\begin{align*}
		\Pi_i=\tilde{\Pi}_i=\{\pm0.1i{\rm j}\}
	\end{align*}
	Therefore all the transmission zeros are on the imaginary axis and, set $\rho_i=(1-0.1i)/2$.
	
	Although the output of each agent is one dimensional, i.e. $q=1$, we define the dimension of the distributed exosystem generator as $r=2$ and for the agent $i\in\mathcal{V}_r$, set
	\begin{align*}
		&S_i(0) = \left[\begin{array}{cc}0 & 1\\ -1 & 0 \end{array}\right]\\
		&\hat{\beta}_i(0) = 1.
	\end{align*}
	Then we set $S_i(0)=0_{2\times2}$ $(i\in\mathcal{V}\setminus\mathcal{V}_r)$ and choose $\hat{\beta}_i(0)(i \in\mathcal{V}\setminus\mathcal{V}_r)$ uniformly at random from the interval $[-1,1]$ and $w_i(0)$, $x_i(0)$, $\xi_i(0)$ ($i\in\mathcal{V}$) uniformly at random from the interval $[-1,1]$, and set $Q_i(0)=\left[-i\ 0\right]$ for all $i\in\mathcal{V}$.
	Then we apply the distributed exosystem generator (\ref{eq:DEG_synch}) and the distributed dynamic compensator (\ref{eq:compensator_synch}), (\ref{eq:Q_consensus_synch}).
	We choose the matrix $K_i(t)$ such that the eigenvalues of (\ref{eq:stablematrix}) are $\left\{-0.7, -0.8, -0.9, -1.0\right\}$.

	The simulation result is displayed in Fig.~\ref{fig:result_synch2}.
	This figure shows the outputs $y_i$ of all agents. Observe that all outputs synchronize.
	This example illustrates the effectiveness of our proposed controller for achieving robust output synchronization.
	\begin{figure}[t]
		\centering
		\includegraphics[scale=1.0]{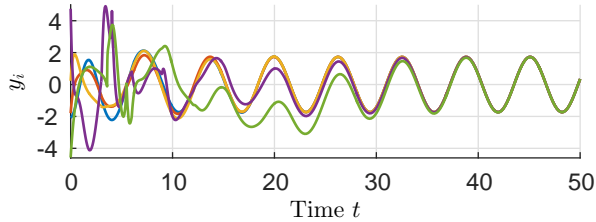}
		\caption{Simulation result (output synchronization problem with multiple roots)}
		\label{fig:result_synch2}
	\end{figure}

	\section{CONCLUSIONS}
\label{sec:conclusion}
	We have studied a multi-agent output regulation problem, where the (linear) agents are heterogeneous, subject to parameter uncertainty, and the network is time-varying. The challenge is that the
	exosystem's dynamics is not accessible by all agents, and consequently the agents do not initially possess a precise internal model of the exosystem. We have solved the problem by proposing
	a distributed controller consisting of two parts -- an exosystem generator that ``learns'' the dynamics of the exosystem and a dynamic compensator that ``learns'' the internal model.
	Moreover, we have extended this controller to solve a related problem, output synchronization, in which there is no exosystem.
	The effectiveness of our solution suggests a {\it distributed internal model principle}: converging internal models imply network output regulation/synchronization.

	\bibliographystyle{IEEEtran}
	\bibliography{mybib}

	\section*{APPENDIX}
	{\noindent\hspace{2em}{\itshape Proof of Lemma~\ref{lem:converge}: }}
		Since the origin is a uniformly exponentially stable equilibrium of $\dot{x} = A_1(t) x$, there exist bounded and positive definite matrices
		$P_1(t), Q_1(t)$ (for all $t \geq 0$) such that 
		\begin{align*}
			\dot{P}_1(t) + P_1(t) A_1(t) + A_1(t)^\top P_1(t) = -Q_1(t).
		\end{align*}
		Then $V_1(x,t) := x^\top P_1(t) x$ is a quadratic Lyapunov function for $\dot{x} = A_1(t)x$, and there exist constants $c_1,c_2,c_3,c_4$
		such that the following are satisfied (globally):
		\begin{align*}
			& c_1 ||x||^2 \leq V_1(x,t) \leq c_2 ||x||^2 \\
			& \frac{\partial V_1}{\partial t} + \frac{\partial V_1}{\partial x} A_1(t) x \leq -c_3 ||x||^2 \\
			& ||\frac{\partial V_1}{\partial x}|| \leq c_4 ||x||.
		\end{align*}
		Now consider $\dot{x} = A_1(t)x + A_2(t)x$. The term $A_2(t)x$ satisfies the inequality
		\begin{align*}
			|| A_2(t)x || \leq || A_2(t) || \cdot ||x||.
		\end{align*}
		Since $A_2(t) \rightarrow 0$, we have $|| A_2(t) || \rightarrow 0$.
		Hence viewing $A_2(t)x$ as a vanishing perturbation to $\dot{x} = A_1(t)x$, it follows from \cite[Corollary~9.1 and Lemma~9.5]{Kha:02} that
		the origin is also a uniformly exponentially stable equilibrium of $\dot{x} = A_1(t)x + A_2(t)x$. In turn, there exist bounded and positive
		definite matrices $P_2(t), Q_2(t)$ (for all $t \geq 0$) such that \begin{align*}
		\dot{P}_2(t) &+ P_2(t) (A_1(t)+A_2(t)) \\
			&+ (A_1(t)+A_2(t))^\top P_2(t) = -Q_2(t).
		\end{align*}
		Let $V_2(x,t) := x^\top P_2(t) x$ be a candidate Lyapunov function for (\ref{eq:A1A2A3}). Then 
		\begin{align*}
			\frac{\partial V_2}{\partial t} &+ \frac{\partial V_2}{\partial x} ( A_1(t)x + A_2(t)x + A_3(t) ) \\
			&= -x^\top Q_2(t) x + 2 x^\top P_2(t) A_3(t) x\\
			&\leq -(||Q_2(t)||-\frac{1}{\epsilon})||x||^2 + \epsilon ||P_2(t)A_3(t)||^2 \\
			&\leq -(||Q_2(t)||-\frac{1}{\epsilon})||x||^2 + \epsilon ||P_2(t)||^2 ||A_3(t)||^2.
		\end{align*}
		Let $\epsilon$ be such that $\epsilon >0$ and $||Q_2(t)||-\frac{1}{\epsilon} >0$. Then it follows from \cite[Theorem~5]{Son:07} that (\ref{eq:A1A2A3})
		is input-to-state stable, with $A_3(t)$ the input. Since $A_3(t) \rightarrow 0$ (uniformly exponentially), as a consequence of input-to-state stability
		(\cite[Section~3.1]{Son:07}, \cite[Section~4.9]{Kha:02}) we conclude that $x(t) \rightarrow 0$ (uniformly exponentially) as $t \rightarrow \infty$.
	{\hspace*{\fill}~\QED\par\endtrivlist\unskip}
	\bigskip	
	{\noindent\hspace{2em}{\itshape Proof of Lemma~\ref{lem:Swconverge}: }}
		From \cite[Theorem~1]{moreau2004stability} and (\ref{eq:Strans}),
		$S_i(i\in\hat{\mathcal{V}})$ reach consensus for all $S_i(0)$.		
		Since Assumption~\ref{ass:spanningtree} holds,
		the consensus value is $S_0$, i.e. $S_i(t)\rightarrow S_0$.
		
		To show $w_i\rightarrow w_0$, we consider
		\begin{align}
			\dot{w}_i=S_0w_i+\sum_{j=0}^N a_{ij}(t)(w_j-w_i)\label{eq:Provisional_cform}.
		\end{align}
		From the proof of Lemma 1 of \cite{scardovi2009synchronization} and by Assumption~\ref{ass:spanningtree}, $w_i\rightarrow w_0$ as $t\rightarrow\infty$.
		Using $\tilde{w}_i\coloneqq w_i-w_0$, we derive
		\begin{align}
			\dot{\tilde{w}}_i&=S_0w_i+\sum_{j=0}^N a_{ij}(t)(w_j-w_0-w_i+w_0)-S_0w_0\nonumber\\
			&=S_0\tilde{w}_i+\sum_{j=0}^N a_{ij}(t)(\tilde{w}_j-\tilde{w}_i).\label{eq:firstalgorithm}
		\end{align}
		Note that $\tilde{w}_0(t)=0$ for all $t\geq0$.
		We define $L^-(t)\in\mathbb{R}^{N\times N}$ to be the matrix obtained by removing the first row and the first column of graph Laplacian $L(t)$.
		In a compact form, (\ref{eq:firstalgorithm}) can be written as
		\begin{align}
			\dot{\tilde{w}}&=(I_N\otimes S_0-L^-(t)\otimes I_r)\tilde{w}\label{eq:matrixformfirst}
		\end{align}
		where $\tilde{w}\coloneqq[\tilde{w}_1^\top\cdots\tilde{w}_N^\top]^\top$.
		Since $\tilde{w}\rightarrow0$, the origin is a uniformly exponentially stable equilibrium of (\ref{eq:matrixformfirst}).
		
		Now we consider (\ref{eq:exoest}). Using $\tilde{S}_i\coloneqq S_i-S_0$,
		\begin{align*}
			\dot{\tilde{w}}_i&=S_i(t)w_i+\sum_{j=0}^N a_{ij}(t)(w_j-w_i)-S_0w_0\nonumber\\
			&=S_0\tilde{w}_i+\tilde{S}_iw_i+\sum_{j=0}^N a_{ij}(t)(\tilde{w}_j-\tilde{w}_i)\nonumber\\
			&=S_0\tilde{w}_i+\tilde{S}_i\tilde{w}_i+\tilde{S}_iw_0+\sum_{j=0}^N a_{ij}(t)(\tilde{w}_j-\tilde{w}_i)
		\end{align*}		
		and in a compact form,
		\begin{align*}
			\dot{\tilde{w}}=&(I_N\otimes S_0-L^-(t)\otimes I_r)\tilde{w}+{\rm diag}(\tilde{S}_1,\dots,\tilde{S}_N)\tilde{w}\nonumber\\
			&+{\rm diag}(\tilde{S}_1,\dots,\tilde{S}_N)({\bf 1}_N\otimes w0)
		\end{align*}
		where $\tilde{w}\coloneqq[\tilde{w}_1^\top\cdots\tilde{w}_N^\top]^\top$.
		Since $\tilde{S}_i\rightarrow0$, $\tilde{w}\rightarrow0$ from Lemma~\ref{lem:converge}.
		Therefore $w_i\rightarrow w_0$ as $t\rightarrow\infty$ for all $i\in\mathcal{V}$.
	{\hspace*{\fill}~\QED\par\endtrivlist\unskip}
	\bigskip	
	{\noindent\hspace{2em}{\itshape Proof of Lemma~\ref{lem:Swconverge_synch}: }}
		Without loss of generality, we reorder the index of agents as $\mathcal{V}_r=\{1,\dots,k\}$ and $\mathcal{V}\setminus\mathcal{V}_r=\{k+1,\dots,N\}$.
		Let $\bar{S}=[S_1(t)^\top\cdots S_N(t)^\top]^\top\in\mathbb{R}^{Nr\times r}$ be a bundled variable.
		In a compact form with respect to $S_i$, (\ref{eq:DEG_synch}) can be written as
		\begin{align*}
			&\dot{\bar{S}}=-(L_r(t)\otimes I_r)\bar{S},\\
			&\bar{S}(0)=[S^{*^\top}\cdots S^{*^\top}\ S_{k+1}(0)^\top\cdots S_N(0)^\top]^\top
		\end{align*}
		where 
		\begin{align*}
			L_r(t) = \left[\begin{array}{c|c}L_1(t) & 0 \\\hline L_2(t) & L_3(t)\end{array}\right].
		\end{align*}
		From \cite[Theorem~1]{moreau2004stability} and Assumption~\ref{ass:underlyingdigraph}, every $S_i (i \in \mathcal{V}_r)$ is such that $S_i(t) = S^*$ for all $t\geq0$, and every $S_i (i \in \mathcal{V} \setminus \mathcal{V}_r)$ is such that $S_i(t)\rightarrow S^*$ as $t\to\infty$. Therefore every $S_i (i \in \mathcal{V})$ reaches consensus for $S_i(0)=S^* (i \in \mathcal{V}_r)$ and arbitrary $S_i(0) (i \in \mathcal{V} \setminus \mathcal{V}_r)$, and the consensus value is $S^*$.
		
		To show $(w_i-w_j)\rightarrow 0$, we again consider (\ref{eq:Provisional_cform}) in the proof of Lemma~\ref{lem:Swconverge} above. From the proof of Lemma 1 of \cite{scardovi2009synchronization} and by Assumption~\ref{ass:underlyingdigraph}, $w_i\rightarrow w^*$ as $t\rightarrow\infty$ for all $i \in\mathcal{V}$. Here $w^*$ is a (virtual) signal generated by $\dot{w}^* = S^*w^*$ and $w^*(0)$ is related only to $w_i(0), i\in\mathcal{V}_r$.
		
		As with the proof of Lemma~\ref{lem:Swconverge} using $w^*$ instead of $w_0$, it is again derived that  $w_i\rightarrow w^*$ i.e. $(w_i-w_j)\rightarrow 0$ for all $i,j\in\mathcal{V}$.
	{\hspace*{\fill}~\QED\par\endtrivlist\unskip}

\addtolength{\textheight}{-12cm}   
\end{document}